\documentclass[12pt]{article}
\usepackage{amsmath, amssymb, amsfonts, amsthm}
\usepackage{mathtools,mathrsfs}
\usepackage{fullpage, setspace}
\usepackage{bm, comment}
\usepackage{booktabs}
\usepackage{natbib}
\usepackage{color}
\usepackage{graphicx} 
\usepackage{changepage}
\usepackage{ulem}
\usepackage{authblk}

\DeclareMathOperator*{\argmax}{arg\,max}

\newtheorem{definition}{Definition}
\newtheorem{proposition}{Proposition}

\definecolor{teal}{RGB}{0,128,128}

\definecolor{purple}{RGB}{128,0,128}

\title{Uncertainty Quantification in Bayesian Clustering}
\date{}

\author[1]{Garritt L. Page}
\author[2]{Andr\'es F. Barrientos}
\author[1]{David B. Dahl} 
\author[3]{\ \\ David B. Dunson}

\affil[1]{Department of Statistics, Brigham Young University, USA }
\affil[2]{Department of Statistics, Florida State University, USA}
\affil[3]{Department of Statistical Science, Duke University, USA}

\begin{document}

\maketitle

\doublespacing

\begin{abstract}
Bayesian clustering methods have the widely touted advantage of providing a probabilistic characterization of uncertainty in clustering through the posterior distribution. An amazing variety of priors and likelihoods have been proposed for clustering in a broad array of settings. There is also a rich literature on Markov chain Monte Carlo (MCMC) algorithms for sampling from posterior clustering distributions. However, there is relatively little work on summarizing the posterior uncertainty. The complexity of the partition space corresponding to different clusterings makes this problem challenging. We propose a post-processing procedure for any Bayesian clustering model with posterior samples that generates a credible set that is easy to use, fast to compute, and intuitive to interpret. We also provide new measures of clustering uncertainty and show how to compute cluster-specific parameter estimates and credible regions that accumulate a desired posterior probability without having to condition on a partition estimate or employ label-switching techniques. We illustrate our procedure through several applications.
\end{abstract}

\section{Introduction}


Quantifying uncertainty associated with objects that are estimated from data is a key contribution that the field of statistics makes within the scientific community.  Uncertainty quantification becomes more challenging as the objects being estimated increase in complexity.  This is unquestionably true in clustering when the goal of the analysis is to estimate a partition or grouping of the $n$ subjects in a study.   To make ideas concrete, let $\rho = \{C_1, \ldots, C_k\}$ denote a random partition of $n$ units into $k$ clusters (i.e., groups) such that: (i) $\cup_{j=1}^k C_j = \{1, \ldots, n\}$, (ii) $C_j \not= \emptyset$, and (iii) $C_j \cap C_{j'} = \emptyset$ for $j \ne j'$. Let $\mathcal{P}$ be the space where $\rho$ takes values, that is, the set of all possible partitions of $n$ items. This space is discrete and finite but suffers from combinatorial explosion resulting in its size growing exponentially as $n$ increases.   Additionally, the relationships among partitions are nonlinear and there is not a natural ordering of partitions, so the meaning of relationships between different partitions is not straightforward and there is little intuition regarding distances among partitions. As a result, estimating the ``central partition'' is not obvious, and assessing uncertainty is even more challenging.

From a Bayesian perspective, the uncertainty associated with a parameter is provided by the posterior distribution.  The more diffuse the posterior, the more uncertainty that accompanies it.   For parameters in $\mathbb{R}^n$, it is common to summarize a parameter's uncertainty by constructing an interval or region from its posterior distribution containing pre-specified posterior mass.  The width of the interval (or volume of the region) gives an indication of the uncertainty, with wider intervals (or larger volumes) indicating more uncertainty.  Unfortunately, such an approach is not straightforward for partitions as they do not belong to $\mathbb{R}^n$.  For a highly concentrated posterior partition distribution, it may be meaningful to assess uncertainty by reporting the probability associated with a partition estimate, but this becomes meaningless as the posterior distribution becomes more diffuse.  

\cite{wade&ghahramani:2018} defined credible balls over $\mathcal{P}$ based on
\begin{align}\label{eq:wade_ball}
\mathcal{B}_{\epsilon}(r_0) = \left\{ r \in \mathcal{P}  \, : \, d(r, r_0) \le \epsilon \right\}, 
\end{align}
where $r_0$ is a fixed partition, $d(\cdot, \cdot)$ is a metric (such as Binder's loss or variation of information) and $\epsilon$ is selected so that $\Pr(\rho \in \mathcal{B}_{\epsilon}(r_0) \mid \bm{y}) \approx \gamma$ for some pre-specified $\gamma \in [0,1]$.  Although this approach does define a collection of partitions that have the required posterior mass, the interpretation of $\mathcal{B}_{\epsilon}(r_0)$ is opaque, since it is not clear how $r \in \mathcal{B}_{\epsilon}(r_0)$ are related other than not being greater than $2\epsilon$ from each other in the chosen loss, which has an unclear meaning in practice. To more succinctly summarize $\mathcal{B}_{\epsilon}(r_0)$, \cite{wade&ghahramani:2018} suggest reporting the vertical upper bound (a partition in $\mathcal{B}_{\epsilon}(r_0)$ with the most clusters that is farthest from $r_0$), the lower vertical bound (a partition in $\mathcal{B}_{\epsilon}(r_0)$ with the fewest clusters that is farthest from $r_0$), and the horizontal bound (a partition in $\mathcal{B}_{\epsilon}(r_0)$ that is farthest from $r_0$).  Of course, for a given summary, multiple partitions could satisfy the definition.  Furthermore, due to the complexity of $\mathcal{P}$, the three boundary partitions are often not very enlightening.  This is even more acutely true when the available software for finding the boundary partitions only considers partitions visited in an MCMC sample.

Rather than relying on an $\epsilon$ loss threshold, we construct a subset of $\mathcal{P}$ whose elements all share a specific criterion: all partitions in our subset have a common subset of individuals that are clustered in the same way.  This is done by formulating a sequence of subsets of $\mathcal{P}$, starting with $\mathcal{P}$ and ending with a single partition $\{r_0\}$. This framework enhances interpretability as each element of the sequence of subsets is composed of a  subpartition for a subset of the $n$ units. The general idea that motivated considering these subpartitions is tangentially related to the framework described in \cite{barrientos_etal:2022} who developed a method of quantifying uncertainty in ranking or ordering-related statements.   Like partitions, a ranking of $n$ units lives in a complex space and is typically accompanied by high, yet difficult-to-assess uncertainty.

The constructed sequence of subpartitions results in innovations to inference in partition modeling.  First, we produce a credible set (that is, a subset of $\mathcal{P}$) that is very interpretable, since all partitions in the credible set share a common subpartition. As a result, we can say something concretely about all the partitions in the credible set.  Second, the credible set is easily computed using new functionality in established software ({\tt salso} {\tt R} package \citealt{dahl_etal:2021}). Third, the credible set is intuitive in that more uncertainty is reflected as the posterior distribution becomes more diffuse.  Fourth, we are able to provide an overall measure of posterior uncertainty comparable to Euclidean concepts such as standard deviation. Fifth, we provide an interpretable unit-level measure of uncertainty for observations outside the subpartition based on posterior probabilities. Finally, an advantage of our approach is that it provides cluster-specific parameter estimates and credible regions that accumulate a desired posterior probability without having to condition on a partition estimate or employ relabeling techniques.  All of this is available as a post-processing procedure that is generally applicable and not tied to any particular clustering model.  

A disclaimer is that our approach for summarizing the posterior distribution of $\rho$ is based on posterior samples. We implicitly assume that these samples provide a good representation of the true posterior, which in turn provides an accurate characterization of our uncertainty in the inferred clustering structure in the data. If a poor model is chosen, the MCMC sampler does not converge, or the MCMC does not mix adequately, then all approaches for summarizing the posterior distribution from these samples, including ours, will produce misleading results.

 \cite{buch2024bayesianlevelsetclustering}  estimate partitions based on a density estimate, resulting in a subpartition (only units with high density are included in the clustering) that is used to form $\mathcal{B}_{\epsilon}(r_0)$.  Unfortunately, this ball is difficult to interpret since not all partitions in $\mathcal{B}_{\epsilon}(r_0)$ contain the subpartition.  To better summarize a credible set, \cite{balocchi2025understandinguncertaintybayesiancluster} report a collection of partition point estimates that are analogous to mode centers. This is a useful summary of the posterior distribution, but does not quantify its spread or variability.  
\cite{lavigne&liverani:2024} compute the predictive probability that each point belongs to a cluster based on a point estimate. \cite{zhu&melnykov:2015} provide an approach from a frequentist perspective, treating posterior probabilities of cluster membership as functions of mixture model parameters and using the delta method to derive their standard errors.  Finally, the posterior pairwise co-clustering matrix has been used as a tool to assess uncertainty, but concrete information gleaned from this tool is limited. 

The remainder of the paper is organized as follows.  In Section \ref{sec:definition}, we build intuition through a simulated example and then rigorously define our interpretable credible set.  Section \ref{sec:computation} contains details on how the interpretable credible set is computed.  Section \ref{sec:summarizing_uncertainty} shows how the interpretable credible set can be used to assess and quantify partition uncertainty.  Section \ref{sec:illustrationssimulations} contains a simulation study  and two data applications that illustrate our approach.  We provide some concluding remarks in Section \ref{sec:conclusion}.



\section{Defining an Interpretable Credible Set} \label{sec:definition}

\subsection{Synthetic Example to Build Intuition}\label{sec:intuition_building}
We first introduce a simple example to illustrate our proposed method and the inferences that can be made regarding posterior partition uncertainty.  In subsequent sections, we describe our method using precise notation and exact details, but it is helpful to build intuition to guide understanding as heavy notation is encountered in subsequent sections. This example will be discussed more extensively in Section \ref{sec:toy_example}.

\begin{figure}
\begin{center}
\includegraphics[scale=0.8]{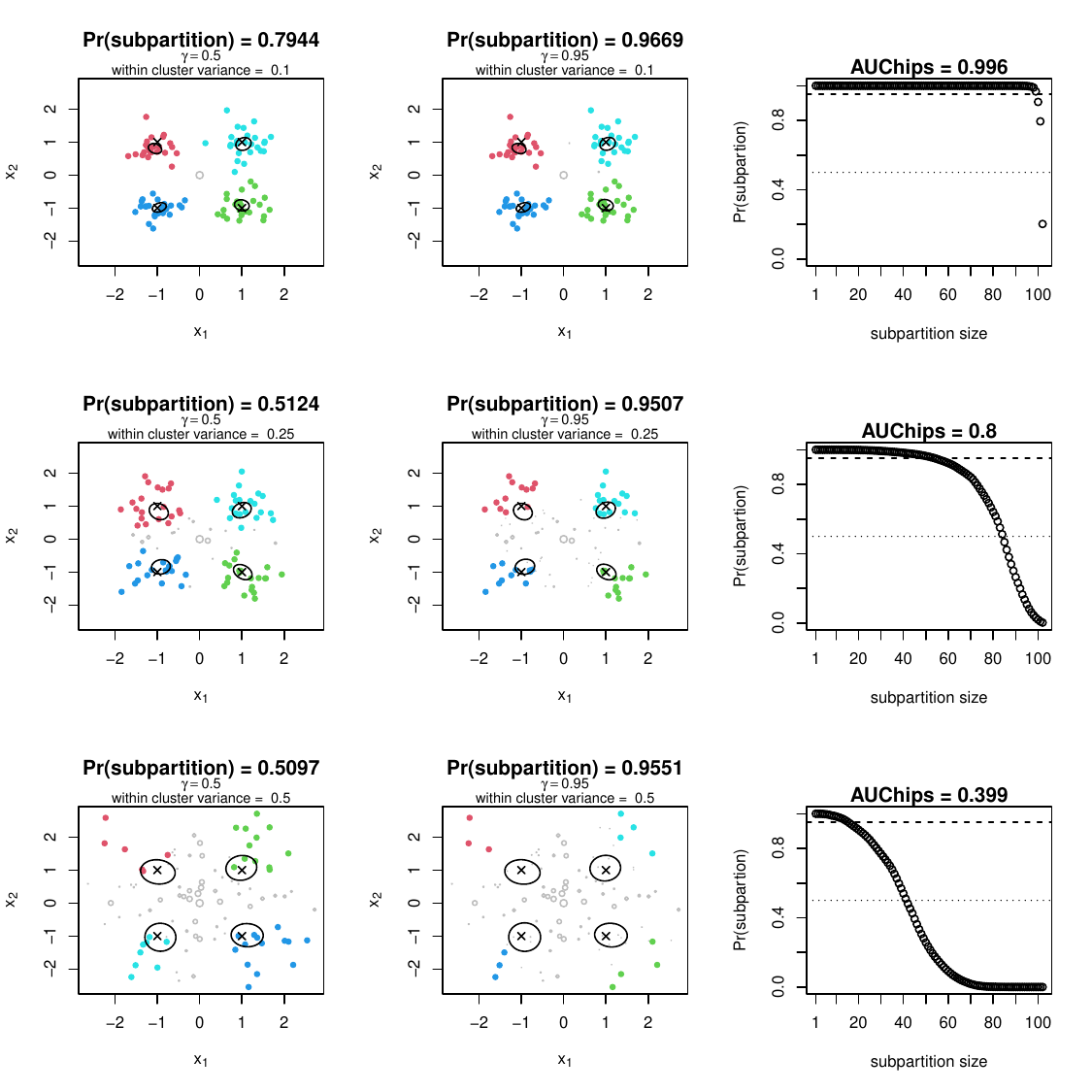}
\end{center}
\caption{Results from a single synthetic dataset generated from a bivariate Gaussian mixture.  Each row corresponds to a different value of between-cluster spread.}
\label{fig:toyexample}
\end{figure}

Consider the simulated datasets displayed in Figure \ref{fig:toyexample} that consist of four clusters each with 25 observations (100 observations in total) generated from a bivariate Gaussian distribution whose mean is one of $\{(-1,-1), (-1,1), (1,-1), (1,1)\}$ and covariance is $\sigma^2 \bm{I}$. The overlap between clusters (i.e., variability within each cluster) is progressively increased by considering the following within cluster variances $\sigma^2 \in \{0.1, 0.25, 0.5\}$. One additional point at $(0,0)$ is added for discussion purposes.

Based on these observations, we sampled cluster labels from the posterior clustering distribution using the MCMC approach outlined in Section \ref{sec:illustrationssimulations}.  Figure \ref{fig:toyexample} illustrates the contributions our approach provides.  The first two columns display the 100 observations and subpartitions (depicted by the colored points) that correspond to two posterior probability thresholds which are labeled by $\gamma=0.5$ and $\gamma= 0.95$. As mentioned in the Introduction,  we use the term subpartition to describe the partition of a subset of individuals. The threshold $\gamma$ is an input parameter that resembles the idea of a credible level. Due to the discreteness of the space $\mathcal{P}$ and the finite number of MCMC samples, it is not guaranteed that we will find a subpartition with exact posterior probability $\gamma$. Instead, we aim to identify a subpartition that accumulates at least $\gamma$ posterior probability while including as many individuals as possible. Once identified, this subpartition is expected in at least $\gamma \times 100\%$ of the posterior draws of $\rho$. The actual posterior probability of the reported subpartition in Figure~\ref{fig:toyexample} is shown in each subfigure above the employed threshold $\gamma$ values.

The credible sets we construct include partitions that contain the respective subpartition. For example, the subpartition in the top plot of the first column contains all points except the middle point at $(0,0)$.  Thus, the subset of $\mathcal{P}$ that we construct contains five partitions, one where the point $(0,0)$ is a singleton cluster, one for which the point $(0,0)$ is assigned to the red cluster, one where the point $(0,0)$ is assigned to the blue cluster, etc.  When an observation such as $(0,0)$ is difficult to assign to a cluster, while maintaining a certain posterior probability, it is preferable not to make a specific statement about how it is clustered, but rather quantify the uncertainty introduced when trying to assign it to a cluster. In the figure, as the posterior probability threshold increases (i.e., moving from the left column to the middle column), the subpartition contains fewer observations.  Additionally, as the cluster overlap diminishes (moving from the top row to the bottom row), the subpartition contains fewer observations.  The gray points do not belong to the subpartitions and their circumference corresponds to point-level uncertainty.  Points whose circumference is large have a high uncertainty in being allocated to a cluster in the subpartition.

The ``$\times$'' symbols in each subfigure of the first two columns of Figure \ref{fig:toyexample} correspond to the true cluster centers used to generate data.  The ellipses correspond to credible regions computed using our approach to carry out cluster-specific parameter inference based on the subpartition.  All ellipses contain the truth (i.e., the ``$\times$'' symbols) except for the dataset in the top row. Notice further that, as the subpartition size decreases, the volume of the ellipses increases.  The credible regions depicted by the ellipses are computed using the joint posterior distribution of the partition $\rho$ and cluster-specific parameters. This is done by using the cluster-specific parameters drawn from the posterior distribution, restricted to those draws of $\rho$ in which the given subpartition is present.  The result is a unified approach that avoids the two-step procedure of estimating the partition first and then estimating cluster-specific parameters for a fixed (but uncertain) partition estimate.

The right column of Figure \ref{fig:toyexample} illustrates our global metric of uncertainty. The $x$-axis represents the number of individuals in the subpartition, and the $y$-axis represents the largest posterior probability across all partitions of that size. As expected, the posterior probability decreases as the number of individuals in the subpartition increases. Notice that less cluster overlap (which implies lower uncertainty in estimating the partition) results in long-lasting high posterior probability and more ``area under the curve.'' We interpret this area under the curve as an overall uncertainty metric and refer to it as AUChips. Values of AUChips that are close to one correspond to a very concentrated posterior distribution of the partition $\rho$, while values that approach 0 correspond to a posterior distribution that is more diffuse.

With this brief demonstration of the insights that our approach can provide, we now turn to the technical details.

\subsection{Notation and Preliminaries}\label{sec:notation}

Below we introduce the notation used to precisely describe our methodology.  Recall that we use $\rho$ to denote a random partition of $\{1,\ldots,n\}$ taking values in $\mathcal{P}$.  We will use $r = \{C_1, \ldots, C_k\} \in \mathcal{P}$ to denote a realization of $\rho$. For a given subset $A \subseteq \{1, \ldots, n\}$, we define $r_{A} = \{C_1 \cap A, \ldots, C_k \cap A\}$ and refer to it as the subpartition of $r$ associated with the subset $A$.  Similarly,  $\rho_A$ will denote the subpartition of $\rho$ associated with the subset $A$.

\subsection{Constructing an Interpretable Subset of $\mathcal{P}$}\label{sec:openset_description}
Let $\boldsymbol{\pi} = (\pi_1, \ldots, \pi_n)$ be a permutation of the vector $(1, \ldots, n)$ and $\boldsymbol{\pi}_{1:\ell}  = \{ \pi_1, \ldots, \pi_\ell\}\subseteq \{1, \ldots, n\}$. While $\boldsymbol{\pi}$ is a vector,  $\boldsymbol{\pi}_{1:\ell} = \{ \pi_1, \ldots, \pi_\ell\}$ is a set, which means that $\boldsymbol{\pi} \neq \boldsymbol{\pi}_{1:n} = \{1, \ldots, n\}$. The set of all possible partitions of $n$ items is $\mathcal{P}$
and we introduce a sequence of subsets of $\mathcal{P}$ that will play a key role in defining neighborhoods of partitions based on shared subpartitions.
\begin{definition}\label{def_open_set_seq}
Given a partition $r_0 \in \mathcal{P}$ and a permutation $\boldsymbol{\pi}$ of $\{1, \ldots, n\}$, define $\mathcal{O}(r_0, \boldsymbol{\pi}_{1:\ell})$ as a sequence of subsets of $\mathcal{P}$ as
\[
\mathcal{O}(r_0, \boldsymbol{\pi}_{1:\ell}) = \left\{ r \in \mathcal{P} : r_{0,\boldsymbol{\pi}_{1:\ell}} \text{ is a subpartition of } r \right\}, \quad \ell = 1, \ldots, n,
\]
where $r_{0,\boldsymbol{\pi}_{1:\ell}}$ denotes the subpartition of $r_0$ associated with the subset $\boldsymbol{\pi}_{1:\ell}$ and $\ell$ denotes the number of units in $r_{0,\boldsymbol{\pi}_{1:\ell}}$. 
\end{definition}
\noindent Each set $\mathcal{O}(r_0, \boldsymbol{\pi}_{1:\ell})$ contains all partitions that share a common subpartition with $r_0$. The sequence $\{\mathcal{O}(r_0, \boldsymbol{\pi}_{1:\ell})\}_{\ell=1}^n$ exhibits a monotone structure that facilitates constructing an interpretable credible set for $\rho$.  This is formalized in the following proposition.
\begin{proposition}[Monotonicity of $\mathcal{O}(r_0,\boldsymbol{\pi}_{1:\ell})$]\label{monotonicity_open_set}
For $\ell=1,\ldots,n$, the sequence $\mathcal{O}(r_0,\boldsymbol{\pi}_{1:\ell})$ is nested and decreasing:
$\mathcal{O}(r_0,\boldsymbol{\pi}_{1:\ell+1}) \subseteq \mathcal{O}(r_0,\boldsymbol{\pi}_{1:\ell})$ for all $1\le\ell<n$. 
Moreover, $\mathcal{O}(r_0,\boldsymbol{\pi}_{1:1})=\mathcal{P}$ and $\mathcal{O}(r_0,\boldsymbol{\pi}_{1:n})=\{r_0\}$.
\end{proposition}

\begin{proof}
If $r \in \mathcal{O}(r_0,\boldsymbol{\pi}_{1:\ell+1})$, then by definition 
$r_{\boldsymbol{\pi}_{1:\ell+1}} = r_{0,\boldsymbol{\pi}_{1:\ell+1}}$.  
Since $\boldsymbol{\pi}_{1:\ell} \subset \boldsymbol{\pi}_{1:\ell+1}$ for $1 \le \ell < n$, 
restricting both sides to $\boldsymbol{\pi}_{1:\ell}$ gives 
$(r_{\boldsymbol{\pi}_{1:\ell+1}})|_{\boldsymbol{\pi}_{1:\ell}} = r_{\boldsymbol{\pi}_{1:\ell}} = r_{0,\boldsymbol{\pi}_{1:\ell}}= (r_{0,\boldsymbol{\pi}_{1:\ell+1}})|_{\boldsymbol{\pi}_{1:\ell}}$, 
where $r|_{S}$ denotes the restriction of $r$ to the subset $S$.  
Hence $r \in \mathcal{O}(r_0,\boldsymbol{\pi}_{1:\ell})$, proving that
$
\mathcal{O}(r_0,\boldsymbol{\pi}_{1:\ell+1}) \subseteq \mathcal{O}(r_0,\boldsymbol{\pi}_{1:\ell}).
$
For the endpoints, when $\boldsymbol{\pi}_{1:1} = \{\pi_1\}$, any partition restricted to a singleton coincides with $r_0$ on that element, so $\mathcal{O}(r_0,\boldsymbol{\pi}_{1:1}) = \mathcal{P}$.  
When $\boldsymbol{\pi}_{1:n} = \{1,\ldots,n\}$, the  condition $r_{\boldsymbol{\pi}_{1:n}} = r_{0,\boldsymbol{\pi}_{1:n}}$ implies $r = r_0$, and thus $\mathcal{O}(r_0,\boldsymbol{\pi}_{1:n}) = \{r_0\}$.
\end{proof}

The nested and monotonic structure of $\{\mathcal{O}(r_0,\boldsymbol{\pi}_{1:\ell})\}_{\ell=1}^n$  induces a monotonic behavior in the corresponding posterior probabilities, as established in the following proposition.
\begin{proposition}[Monotonicity of posterior probabilities over $\mathcal{O}(r_0,\boldsymbol{\pi}_{1:\ell})$]\label{monotonicity_prob}
For any posterior on $\rho$,
\[
\Pr\bigl(\rho\in\mathcal{O}(r_0,\boldsymbol{\pi}_{1:\ell})\mid \bm{y}\bigr)
\text{ is nonincreasing in }\ell,
\]
with 
$\Pr(\rho\in\mathcal{O}(r_0,\boldsymbol{\pi}_{1:1})\mid\bm{y})=1$ 
and 
$\Pr(\rho\in\mathcal{O}(r_0,\boldsymbol{\pi}_{1:n})\mid\bm{y})=\Pr(\rho=r_0\mid\bm{y})$.
\end{proposition}

\begin{proof}
The result follows from the nested nature of the sequence $\{\mathcal{O}(r_0,\boldsymbol{\pi}_{1:\ell})\}_{\ell=1}^n$ established in Proposition~\ref{monotonicity_open_set}. If $A \subseteq B$, then $\Pr(A\mid\bm{y}) \le \Pr(B\mid\bm{y})$. 
Applying this with 
$A=\mathcal{O}(r_0,\boldsymbol{\pi}_{1:\ell+1})$ 
and 
$B=\mathcal{O}(r_0,\boldsymbol{\pi}_{1:\ell})$ 
establishes the probabilistic monotonicity in $\ell$.  The endpoint identities given in Proposition~\ref{monotonicity_open_set} directly imply the stated limits.
\end{proof}

From Propositions~\ref{monotonicity_open_set} and \ref{monotonicity_prob}, the role of $\ell$  in $\mathcal{O}(r_0, \boldsymbol{\pi}_{1:\ell})$ is analogous to that of $\epsilon$ in $\mathcal{B}_{\epsilon}(r_0)$ defined in \eqref{eq:wade_ball}. As $\epsilon \rightarrow 0$, $\mathcal{B}_{\epsilon}(r_0) \rightarrow \{r_0\}$,  while as $\epsilon \rightarrow \infty$, $\mathcal{B}_{\epsilon}(r_0) \rightarrow \mathcal{P}$.  Thus, despite the interpretability challenges associated with $\mathcal{B}_{\epsilon}(r_0)$,  $\epsilon$ provides a straightforward way to connect the size of $\mathcal{B}_{\epsilon}(r_0)$ to a given amount of posterior probability since $\lim_{\epsilon \rightarrow 0} \Pr(\rho \in \mathcal{B}_{\epsilon}(r_0) \mid \bm{y}) = \Pr(\rho = r_0 \mid \bm{y})$ and $\lim_{\epsilon \rightarrow \infty} \Pr(\rho \in \mathcal{B}_{\epsilon}(r_0) \mid \bm{y}) = \Pr(\rho \in \mathcal{P} \mid \bm{y}) = 1$. For our proposed sequence, as stated in Proposition \ref{monotonicity_open_set}, if $\ell = 1$ then $\mathcal{O}(r_0, \boldsymbol{\pi}_{1:1}) = \mathcal{P}$, and as $\ell \rightarrow n$ we have $\mathcal{O}(r_0, \boldsymbol{\pi}_{1:\ell}) \rightarrow \{r_0\}$.  Thus,  increasing $\epsilon$ and decreasing $\ell$ both result in a larger subset of $\mathcal{P}$.  In a similar fashion, Proposition \ref{monotonicity_prob} establishes probability results that are similar to the limiting results associated with $\mathcal{B}_{\epsilon}(r_0)$.  Mainly, the posterior probability $\Pr(\rho \in \mathcal{O}(r_0, \boldsymbol{\pi}_{1:\ell}) \mid \bm{y})$ is monotonically decreasing from $\Pr(\rho \in \mathcal{O}(r_0, \boldsymbol{\pi}_{1:\ell}) \mid \bm{y})=1$ for $\ell=1$ to $\Pr(\rho = r_0 \mid \bm{y})$ for $\ell=n$. Our approach exploits this monotonic behavior and allows us to choose an $\ell$ such that $\Pr(\rho \in \mathcal{O}(r_0, \boldsymbol{\pi}_{1:\ell}) \mid \bm{y})  \ge \gamma$. This will produce a credible region for $\rho$ that is easily interpretable because all partitions have a subpartition in common.

Since any $(r_0, \boldsymbol{\pi})$ can be used to construct a credible region that accumulates at least a desired posterior probability $\gamma$, we next describe a strategy of selecting a credible region in a principled way.  Given the monotone behavior of $\{\mathcal{O}(r_0, \boldsymbol{\pi}_{1:\ell})\}_{\ell=1}^n$ stated in Propositions~\ref{monotonicity_open_set} and \ref{monotonicity_prob}, we know that for any $r_0$ and $\boldsymbol{\pi}$ there exists $n_0$ such that $\Pr(\rho \in \mathcal{O}(r_0, \boldsymbol{\pi}_{1:\ell}) \mid \bm{y}) \geq \gamma$ for all $\ell \leq n_0$. We use $\mathcal{O}(r_0, \boldsymbol{\pi}_{1:n_0})$ rather than any $\mathcal{O}(r_0, \boldsymbol{\pi}_{1:\ell})$ with $\ell < n_0$, because its defining subpartition contains the largest number of items. Among all possible triples $(r_0, \boldsymbol{\pi}, n_0)$, we first select those with the largest possible $n_0$, and then choose the pair $(r_0, \boldsymbol{\pi})$ that yields the subset $\mathcal{O}(r_0, \boldsymbol{\pi}_{1:n_0})$ with the highest posterior probability. Because this construction conditions on $n_0$ and, given $n_0$, selects the subset with the highest inclusion probability, we refer to it as the conditional high-inclusion probability subset (CHIPS) and the CHIPS credible region. The formal definition is given below.

\begin{definition}[CHIPS credible region for $\rho$] \label{def:HPD}
For a given credibility level $\gamma \in (0,1)$, let 
\begin{align}\label{eq:n0}
n_0 = \max\left\{\ell : \Pr(\rho \in \mathcal{O}(r,\tilde{\boldsymbol{\pi}}_{1:\ell}) \mid \bm{y}) \ge \gamma, \, r\in\mathcal{P},\, \mbox{$\boldsymbol{\pi}$ permutation of $(1, \ldots, n)$}\right\}.
\end{align}
and choose the pair $(r_0,\boldsymbol{\pi})$ such that
\begin{align} \label{eq:r0_pi}
(r_0,\boldsymbol{\pi}) = 
\arg\max_{(r,\tilde{\boldsymbol{\pi}})} 
\Pr(\rho \in \mathcal{O}(r,\tilde{\boldsymbol{\pi}}_{1:n_0}) \mid \bm{y}).
\end{align}
The resulting set $\mathcal{O}(r_0,\boldsymbol{\pi}_{1:n_0})$ serves as a CHIPS credible region for $\rho$.
\end{definition}

From now on, any reference to $(r_0, \boldsymbol{\pi}, n_0)$, the subpartition that defines them, or the set $\mathcal{O}(r_0,\boldsymbol{\pi}_{1:n_0})$ refers specifically to the quantities defined in Definition \ref{def:HPD}. Computational strategies for selecting $(r_0, \boldsymbol{\pi}, n_0)$ are described in Section \ref{sec:computation}. To develop these strategies, it is useful to recognize that characterizing  $\mathcal{O}(r_0,\boldsymbol{\pi}_{1:n_0})$ is equivalent to identifying the subpartition that defines it, that is, $r_{0,\boldsymbol{\pi}_{1:n_0}}$. We formalize this result in the following proposition.

\begin{proposition}[Invariance under shared subpartition]\label{prop:invariance}
Fix $\boldsymbol{\pi}_{1:n_0}$ and let $r_0 \in \mathcal{P}$. Then $\mathcal{O}(r_0,\boldsymbol{\pi}_{1:n_0})$ is completely determined by the subpartition $r_{0,\boldsymbol{\pi}_{1:n_0}}$. In particular, for any $r \in \mathcal{O}(r_0,\boldsymbol{\pi}_{1:n_0})$,
$
\mathcal{O}(r_0,\boldsymbol{\pi}_{1:n_0}) = \mathcal{O}(r,\boldsymbol{\pi}_{1:n_0}).
$
\end{proposition}

\begin{proof}
By definition,
$
\mathcal{O}(r_0,\boldsymbol{\pi}_{1:n_0}) = \{ q \in \mathcal{P} : q_{\boldsymbol{\pi}_{1:n_0}} = r_{0,\boldsymbol{\pi}_{1:n_0}} \}.
$
If $r \in \mathcal{O}(r_0,\boldsymbol{\pi}_{1:n_0})$, then $r_{\boldsymbol{\pi}_{1:n_0}} = r_{0,\boldsymbol{\pi}_{1:n_0}}$. Take any $q \in \mathcal{O}(r_0,\boldsymbol{\pi}_{1:n_0})$; then $q_{\boldsymbol{\pi}_{1:n_0}} = r_{0,\boldsymbol{\pi}_{1:n_0}} = r_{\boldsymbol{\pi}_{1:n_0}}$, so $q \in \mathcal{O}(r,\boldsymbol{\pi}_{1:n_0})$. Hence $\mathcal{O}(r_0,\boldsymbol{\pi}_{1:n_0}) \subseteq \mathcal{O}(r,\boldsymbol{\pi}_{1:n_0})$. The same argument with $r_0$ and $r$ interchanged gives $\mathcal{O}(r,\boldsymbol{\pi}_{1:n_0}) \subseteq \mathcal{O}(r_0,\boldsymbol{\pi}_{1:n_0})$. Therefore,
$
\mathcal{O}(r_0,\boldsymbol{\pi}_{1:n_0}) = \mathcal{O}(r,\boldsymbol{\pi}_{1:n_0}),
$
which shows that $\mathcal{O}(r_0,\boldsymbol{\pi}_{1:n_0})$ depends only on the subpartition $r_{0,\boldsymbol{\pi}_{1:n_0}}$.
\end{proof}

In addition to finding a sequence of interpretable subsets (each of which has a subpartition in common) that accumulate a given posterior probability $\gamma$, our approach addresses scenarios in which a subset of items cannot be assigned to a specific cluster with high posterior probability. In particular, when certain units cannot be confidently assigned to a single cluster, the posterior probability of the corresponding subset $\mathcal{O}(r_0,\boldsymbol{\pi}_{1:\ell})$ may drop sharply as those units are incorporated. The following proposition illustrates this phenomenon in a simple case.
\begin{proposition}[Effect of uncertain assignments]\label{prop_drop_prob}
Consider $\mathcal{O}(r_0,\boldsymbol{\pi}_{1:\ell})$ with $\ell < n$, and let $r_{0,\boldsymbol{\pi}_{1:\ell}}$ contain $k$ clusters. Suppose that, conditional on $\rho \in \mathcal{O}(r_0,\boldsymbol{\pi}_{1:\ell})$, the $(\ell+1)$st unit has posterior probability $1/k$ of belonging to each of these $k$ clusters. Then
$$
\Pr(\rho \in \mathcal{O}(r_0,\boldsymbol{\pi}_{1:\ell+1}) \mid \bm{y}) = 
\frac{1}{k} \Pr(\rho \in \mathcal{O}(r_0,\boldsymbol{\pi}_{1:\ell}) \mid \bm{y}).
$$
Hence, adding an element with high assignment uncertainty can substantially reduce the posterior probability of $\mathcal{O}(r_0,\boldsymbol{\pi}_{1:\ell+1})$.
\end{proposition}

\begin{proof}
By definition, $\mathcal{O}(r_0,\boldsymbol{\pi}_{1:\ell+1})$ is the subset of $\mathcal{O}(r_0,\boldsymbol{\pi}_{1:\ell})$ in which the $(\ell+1)$st element is assigned to the same cluster as in $r_0$. Conditional on $\rho \in \mathcal{O}(r_0,\boldsymbol{\pi}_{1:\ell})$, the probability that this assignment occurs is $1/k$, since the $(\ell+1)$st element is equally likely to belong to any of the $k$ clusters in $r_{0,\boldsymbol{\pi}_{1:\ell}}$. It follows that
\begin{align*}
\Pr(\rho \in \mathcal{O}(r_0,\boldsymbol{\pi}_{1:\ell+1}) \mid \bm{y})
&= \Pr(\rho \in \mathcal{O}(r_0,\boldsymbol{\pi}_{1:\ell+1}) 
   \mid \rho \in \mathcal{O}(r_0,\boldsymbol{\pi}_{1:\ell}), \bm{y}) \\
&\quad \times \Pr(\rho \in \mathcal{O}(r_0,\boldsymbol{\pi}_{1:\ell}) \mid \bm{y}) \\
&= \frac{1}{k} \Pr(\rho \in \mathcal{O}(r_0,\boldsymbol{\pi}_{1:\ell}) \mid \bm{y}),
\end{align*}
as claimed.
\end{proof}
Proposition~\ref{prop_drop_prob} highlights that adding elements with uncertain cluster assignments can sharply reduce the posterior probability of $\mathcal{O}(r_0,\boldsymbol{\pi}_{1:\ell})$. The main idea of our method is to permit the inclusion of additional elements into $r_{0,\boldsymbol{\pi}_{1:\ell}}$ as long as the posterior probability remains above the threshold $\gamma$, stopping only when further inclusion would reduce it below this level. The example in Figure~\ref{fig:toyexample} illustrates this point.

\section{Computing an Interpretable Credible Set} \label{sec:computation}

To compute the CHIPS credible region $\mathcal{O}(r_0, \boldsymbol{\pi}_{1:n_0})$,  we rely on the monotone sequence of subsets
$\{r\} = \mathcal{O}(r, \boldsymbol{\pi}_{1:n}) \subset
\ldots \subset
\mathcal{O}(r, \boldsymbol{\pi}_{1:1}) = \mathcal{P}$. By 
Proposition~\ref{prop:invariance} each $\mathcal{O}(r, \boldsymbol{\pi}_{1:\ell})$ is uniquely defined through its corresponding subpartition $r_{\boldsymbol{\pi}_{1:\ell}}$, and so the sequence of monotone subsets is completely described by the sequence of subpartitions $r_{\boldsymbol{\pi}_{1:1}}, \ldots ,r_{ \boldsymbol{\pi}_{1:n}}$ that are themselves monotone in the sense that $r_{\boldsymbol{\pi}_{1:\ell}}$ is a subpartition of $r_{\boldsymbol{\pi}_{1:(\ell+1)}}$.  

Reframing the optimization in Equations \eqref{eq:n0}-\eqref{eq:r0_pi} in terms of subpartitions, we propose an algorithm that optimizes over sequences of monotone subpartitions, improving scalability relative to algorithms that consider all partitions. For each partition $r \in \mathcal{P}$ and permutation $\boldsymbol{\pi}$ of $(1,\ldots,n)$, there exists a $\mathcal{O}(r, \boldsymbol{\pi}_{1:n}) \subset
\ldots \subset \mathcal{O}(r, \boldsymbol{\pi}_{1:1})$, giving rise to a very large number of sequences. To bypass the impracticality of an exhaustive search, we propose a greedy algorithm that identifies sequences of subpartitions by iteratively adding observations in a way that maximizes posterior probability. 

The algorithm is initialized with a specific observation, say $i_1$. 
For the remaining $n-1$ observations indexed by $i$, the posterior probabilities of $\{\{i_1, i\}\}$ and $\{\{i_1\}, \{i\}\}$ are calculated. We choose the highest of these posterior probabilities and let $i_2$ denote the corresponding observation; for example, the highest posterior probability subpartition may be  $\{\{i_1\}, \{i_2\}\}$. The process repeats for the remaining $n-2$ observations. At this stage, the algorithm evaluates all possible ways to cluster each observation $i \neq i_1, i_2$ with the current partition $\{\{i_1\}, \{i_2\}\}$. The potential partitions now include $\{\{i_1, i\}, \{i_2\}\}$, $\{\{i_1\}, \{i_2, i\}\}$, and $\{\{i_1\}, \{i_2\}, \{i\}\}$. Again, the algorithm selects the observation and subpartition that maximize the posterior probability. This process continues, adding one observation at a time, until no observations remain. For each starting point $i_1$, the algorithm produces a partition $r^{(i_1)} \in \mathcal{P}$ and a permutation $\boldsymbol{\pi}^{(i_1)}$ of $(1, \ldots, n)$, where the permutation $\boldsymbol{\pi}^{(i_1)}$ records the order in which observations were added. To simplify notation, we denote $r_{\boldsymbol{\pi}_{1:\ell}}^{(i_1)}$ as the subpartition of $r^{(i_1)}$ associated with the subset $\boldsymbol{\pi}_{1:\ell}^{(i_1)}$. Thus, for a given list of starting points $\mathcal{I}$, each $(r^{(i_1)}, \boldsymbol{\pi}^{(i_1)})$ with $i_1 \in \mathcal{I}$ corresponds to a sequence of monotone subpartitions  $r_{\boldsymbol{\pi}_{1:1}}^{(i_1)}, \ldots, r_{\boldsymbol{\pi}_{1:n}}^{(i_1)}$ and, therefore, subsets $\{r^{(i_1)}\} = \mathcal{O}(r^{(i_1)}, \boldsymbol{\pi}_{1:n}^{(i_1)}) \subset \ldots \subset \mathcal{O}(r^{(i_1)}, \boldsymbol{\pi}_{1:1}^{(i_1)}) = \mathcal{P}$.  We estimate $\mathcal{O}(r_0, \boldsymbol{\pi}_{1:n_0}) \subseteq \mathcal{P}$ by approximating the optimization problem in Equations \eqref{eq:n0} - \eqref{eq:r0_pi} as
\begin{equation}\label{eq:solution_n0_r0_pi}
\begin{aligned}
n_0 & \approx \max\left\{\ell : \Pr(\rho \in \mathcal{O}(r^{(i_1)},\boldsymbol{\pi}_{1:\ell}^{(i_1)} \mid \bm{y}) \ge \gamma, \, i_1 \in \mathcal{I} \right\},\\
(r_0,\boldsymbol{\pi}) & \approx
\underset{i_1 \in \mathcal{I}}{\rm argmax} 
\Pr(\rho \in \mathcal{O}(r^{(i_1)},\boldsymbol{\pi}_{1:n_0}^{(i_1)}) \mid \bm{y}).
\end{aligned}
\end{equation}


This algorithm, which is detailed in the supplemental material,  produced the $r_0$ that corresponds to the colored points in Figure \ref{fig:toyexample}. The resulting subset $\mathcal{O}(r_0, \boldsymbol{\pi}_{1:n_0}) \subset \mathcal{P}$ from this algorithm is an approximation of the CHIPS credible set of Definition \ref{def:HPD}.


The number of possible values for the starting point $i_1$ and the number of posterior draws used to estimate probabilities can influence the estimated CHIPS credible region $\mathcal{O}(r_0,\boldsymbol{\pi}_{1:n_0}) \subset \mathcal{P}$. Even though in theory this subset is  unique in almost all cases, when either of these numbers are small, the subpartition returned by the CHIPS algorithm may not be unique. 
Additional discussion, stability checks for the CHIPS algorithm, and practical guidelines are provided in the supplementary material. These results show the algorithm is quite stable as long as sufficiently many starting points and posterior draws are considered.

\section{Quantifying Clustering Uncertainty}\label{sec:summarizing_uncertainty}

In this section, we propose summaries of posterior uncertainty building on the credible set estimation methodology of Section 3.

\subsection{Quantifying Global Clustering Uncertainty}\label{sec:global_partition_uncertainty}
Our first summary characterizes how quickly the posterior probability associated with a subpartition decreases as additional units are added. We first identify, for a given number of units $\ell$, the subpartition clustering exactly $\ell$ units that accumulates the highest posterior probability. We denote this subpartition as $r^{\rm max}_{\boldsymbol{\pi}_{1:\ell}}$, which is computed using
$$
(r^{\rm max}_\ell, \boldsymbol{\pi}_{1:\ell}^{\rm max}) \approx 
\underset{(r^{(i_1)}, \boldsymbol{\pi}^{(i_1)}_{1:\ell})\mbox{ s.t. }{ i_1 \in \mathcal{I}}}{\rm argmax} \left\{ \Pr(\rho \in  \mathcal{O}(r^{(i_1)}, \boldsymbol{\pi}_{1:\ell}^{(i_1)})  ~|~ \bm{y}) \right\},
$$
where $r^{\rm max}_{\boldsymbol{\pi}_{1:\ell}}$ is the subpartition of $r^{\rm max}_\ell$ associated with the subset $\boldsymbol{\pi}_{1:\ell}^{\rm max}$. For each of these subpartitions, we also define the corresponding posterior probabilities,
$$
p^{\rm max}_{\ell} = \Pr\left(\rho_{\boldsymbol{\pi}^{\rm max}_{1:(\ell-1)}} = r^{\rm max}_{\boldsymbol{\pi}_{1:\ell}} ~ | ~ \bm{y}\right), \quad \ell = 1, \ldots, n.
$$
This sequence of probabilities is non-increasing.  If $n_0$ is the largest $\ell \in \{1,\ldots,n\}$ such that $p^{\rm max}_{\ell} \geq \gamma$, then 
$ (r_0, \boldsymbol{\pi}, n_0)$ in Equation \eqref{eq:solution_n0_r0_pi} satisfies $ (r_0, \boldsymbol{\pi}_{1:n_0}) = (r^{\rm max}_{n_0}, \boldsymbol{\pi}_{1:{n_0}}^{\rm max})$.

In scenarios with very low uncertainty, we expect the posterior distribution to concentrate its mass on a single partition or on very similar partitions. In such cases, the sequence of probabilities $p^{\rm max}_{\ell}$ remains nearly constant and close to one. (See the top row of Figure \ref{fig:toyexample}.) In contrast, in high-uncertainty scenarios, posterior probability is spread across a large subset of $\mathcal{P}$, making it difficult to identify subpartitions that cluster a moderate number of individuals while also having high posterior probability. In this case, we expect the sequence $p^{\rm max}_{\ell}$ to decrease rapidly. (See the bottom row of Figure \ref{fig:toyexample}.)

Hence, we propose using the curve $(\ell, p^{\rm max}_{\ell})$ to characterize how the posterior probability of subpartitions behaves as the number of associated items increases. If this ``CHIPS curve'' drops quickly from 1 to 0, this represents high uncertainty as it is not possible to group a moderate to large number of individuals with a high posterior probability. On the other hand, if the probability remains close to one, we can claim that it is possible to cluster a moderate to large number of individuals with high posterior probability.  An illustration of this curve is provided in the third column of Figure \ref{fig:toyexample}.

We call the area under the CHIPS curve ``AUChips'' and use it as a summary of $(\ell, p^{\rm max}_{\ell})$ quantifying  uncertainty, dividing by sample size $n$ to restrict the values to $[0,1].$ If the AUChips value is close to one, this indicates low uncertainty, while values close to zero reflect high uncertainty, which aligns with the patterns observed in Figure \ref{fig:toyexample}. 
Since $\ell$ is discrete, we use a piecewise linear function $(t, p^{\rm max}_{\lceil n t \rceil})$ for $t \in (0,1)$ as a surrogate for $(\ell, p^{\rm max}_{\ell})$ defining AUChips as
$$
{\rm AUChips} = \int_{0}^1 p^{\rm max}_{\lceil n t \rceil} dt.
$$
This formulation provides a principled way to quantify uncertainty in partitioning by capturing the overall decline of the posterior probability across different subpartition sizes. 

\subsection{Quantifying Unit-Level Uncertainty}\label{sec:unit_level_uncertainty}

While AUChips provides a global measure of clustering uncertainty, we quantify in this subsection the uncertainty associated with each individual observation. 
Individuals included in 
$r_{0,\boldsymbol{\pi}_{1:\ell}}$, the subpartition of $r_0$ associated with the subset $\boldsymbol{\pi}_{1:\ell}$, contribute less uncertainty than those not included. Hence, we focus on characterizing the impact
of the latter individuals 
$\{1,\ldots,n\} \setminus \boldsymbol{\pi}_{1:n_0}$  by measuring the cost, in terms of loss of probability, associated with adding each of them to the subpartition.

Specifically, for each $i \in \{1,\ldots,n\} \setminus \boldsymbol{\pi}_{1:n_0}$, if 
there are $k_0$ clusters in the subpartition $r_{0,\boldsymbol{\pi}_{1:n_0}}$, there are $k_0+1$ possibilities: assignment to an existing cluster or allocation to a new singleton cluster.
Consequently, we obtain $k_0+1$ posterior probabilities, each corresponding to one of these assignments. To summarize the impact of adding observation $i$ to the current subpartition, we propose using the maximum of these posterior probabilities, which we denote as $q_i^{\rm max}$ for $i \in \{1,\ldots,n\} \setminus \boldsymbol{\pi}_{1:n_0}$. By construction, $q_i^{\rm max}$ for any $i$ is always less than or equal to $p^{\rm max}_{n_0}$, and the smaller the value of $q_i^{\rm max}$, the greater the uncertainty associated with observation $i$. Recall that $p^{\rm max}_{n_0} = \Pr(\rho \in \mathcal{O}(r_0, \boldsymbol{\pi}_{1:n_0}) ~|~ \bm{y})$. Additionally, the quantity $p^{\rm max}_{n_0} - q_i^{\rm max}$ represents the decrease in probability associated with the addition of $i$ to the subpartition $r_{0,\boldsymbol{\pi}_{1:\ell}}$. 

For any given $i$, both $q_i^{\rm max}$ and $p^{\rm max}_{n_0} - q_i^{\rm max}$ provide intuitive and interpretable measures of individual uncertainty. Unlike existing measures such as entropy, whose units tend to lack intuition,
our proposed measure directly quantifies the effect of incorporating an observation into a partition in terms of posterior probabilities.

\subsection{Quantifying Uncertainty of Cluster-Specific Parameters} \label{sec:cluster_specific_inference}

In conducting model-based clustering, there is commonly interest in inferring not only the partition but also the model parameters within each cluster. To estimate cluster-specific parameters, a typical approach is to first obtain a point estimate of the partition and then estimate parameters based on partition fixed at this estimate (\citealt{McDowell2018DPGP}). A major issue with this approach, apart from ignoring uncertainty in the clustering, is that it requires conditioning on a partition point estimate whose posterior probability is usually close to zero.
Alternatively, multiple approaches have been proposed to align the cluster labels across MCMC samples (\citealt{JSSv069c01}); the aligned samples of cluster-specific parameters are then used to estimate posterior summaries (\citealt{BayesianFMM}). An issue with this strategy is that, in cases with moderate to high posterior uncertainty, it may not be possible to reliably align the clusters across MCMC iterations.

We avoid the need for conditioning on a partition whose posterior probability is usually close to zero or relabeling by building credible regions for cluster-specific parameters conditional on $r_{0,\bm{\pi}_{1:n_0}} = \{C_{0,1}, \ldots, C_{0,k_0}\}$. In this way, the user retains control over the posterior probability of the subpartition and the corresponding credible regions simultaneously.  For example, the user would be able to state with posterior probability 0.95 that a subset of individuals is clustered in a given way (that is, $r_{0,\bm{\pi}_{1:n_0}}$) and that the mean of one of the specific clusters lies within a given credible region. To this end, for $j  = 1, \ldots, k_0$, let $\theta_j$ denote a generic cluster-specific model parameter whose support is $\Theta \subseteq \mathbb{R}^p$. For example, in the bottom left plot of Figure \ref{fig:toyexample}, $k_0=4$ and $\theta_j$ for $j=1,\ldots,4$  corresponds to a mean parameter belonging to $\mathbb{R}^2$.  We estimate the parameter $\theta_j$ using the observations from units that belong to some cluster $C \in \rho$, where $\rho \in \mathcal{O}(r_0, \boldsymbol{\pi}_{1:n_0}) \subset \mathcal{P}$ and $C_{0,j} \subseteq C \in \rho$ for $C_{0,j} \in r_{0,\bm{\pi}_{1:n_0}}$.  Critically, $\theta_j$ is estimated using not just units that belong to $C_{0,j}$, but all those that belong to $C$.  While $C_{0,j}$ represents a fixed cluster within the subpartition, the cluster $C$ is random.  Thus, for example, the units colored red in the bottom left plot of Figure \ref{fig:toyexample} are only a subset of the observations used to estimate uncertainty about the estimated mean value of the red cluster.  Hence, in estimating $\theta_j$, the uncertainty quantified by the credible region reflects both the randomness of the cluster $C$ and the variability of the observations it contains. 

Thus, our inferential procedure for $\theta_j$ is based on the posterior distribution conditional on $\rho \in \mathcal{O}(r_0, \boldsymbol{\pi}_{1:n_0})$, whose density can be expressed as
\begin{eqnarray}\label{eq:prob_theta_j_0}
&  & \hspace{-15mm} 
p\left(\theta_j ~|  ~ \rho \in \mathcal{O}(r_0, \boldsymbol{\pi}_{1:n_0}), ~ \bm{y} \right) 
\\\nonumber & = & \frac{p\left(\theta_{j},\rho\in\mathcal{O}(r_{0},\boldsymbol{\pi}_{1:n_{0}})|\bm{y}\right)}{p\left(\rho\in\mathcal{O}(r_{0},\boldsymbol{\pi}_{1:n_{0}})|\bm{y}\right)}
\\\nonumber & = & \sum_{r\in\mathcal{O}(r_{0},\boldsymbol{\pi}_{1:n_{0}})}\frac{p\left(\theta_{j},\rho=r|\bm{y}\right)}{p\left(\rho\in\mathcal{O}(r_{0},\boldsymbol{\pi}_{1:n_{0}})|\bm{y}\right)}\\\nonumber
 & = & \sum_{r\in\mathcal{O}(r_{0},\boldsymbol{\pi}_{1:n_{0}})}\frac{p\left(\theta_{j}|\rho=r,\bm{y}\right)p\left(\rho=r|\bm{y}\right)}{p\left(\rho\in\mathcal{O}(r_{0},\boldsymbol{\pi}_{1:n_{0}})|\bm{y}\right)}\\\label{eq:prob_theta_j_1}
 & = & \sum_{r\in\mathcal{O}(r_{0},\boldsymbol{\pi}_{1:n_{0}})}p\left(\theta_{j}|\rho=r,\bm{y}\right)p\left(\rho=r|\rho\in\mathcal{O}(r_{0},\boldsymbol{\pi}_{1:n_{0}}),\bm{y}\right)\\ \label{eq:prob_theta_j_2}
 & = & \sum_{r\in\mathcal{O}(r_{0},\boldsymbol{\pi}_{1:n_{0}})}p\left(\theta_{j}|C_{0,j}\subseteq C\in \rho, \rho=r, \bm{y}\right)p\left(\rho=r|\rho\in\mathcal{O}(r_{0},\boldsymbol{\pi}_{1:n_{0}}),\bm{y}\right)\\
 \label{eq:prob_theta_j_3}
 & = & \sum_{r\in\mathcal{O}(r_{0},\boldsymbol{\pi}_{1:n_{0}})}p\left(\theta_{j}|y_{i},~i\in C,~C_{0,j}\subseteq C\in \rho, \rho=r\right)p\left(\rho=r|\rho\in\mathcal{O}(r_{0},\boldsymbol{\pi}_{1:n_{0}}),\bm{y}\right)
\end{eqnarray}
The transition from $\eqref{eq:prob_theta_j_0}$ to $\eqref{eq:prob_theta_j_1}$ follows directly from the definition of conditional density. The step from $\eqref{eq:prob_theta_j_1}$ to $\eqref{eq:prob_theta_j_2}$ relies on the fact that $\theta_j$ depends only on the cluster in $\rho$ that contains the observations in $C_{0,j}$, denoted by $C$. Once we condition on $\rho$, the posterior distribution of $\theta_j$ depends only on the observations belonging to this random cluster $C$. The transition from $\eqref{eq:prob_theta_j_2}$ to $\eqref{eq:prob_theta_j_3}$ follows from recognizing that $\theta_j$ is informed using observations from the cluster $C$ that contains $C_{0,j}$. The statements in $\eqref{eq:prob_theta_j_0}$–$\eqref{eq:prob_theta_j_3}$ further highlight that $\theta_j$ is not estimated solely from the observations in $C_{0,j}$ but also incorporates information from all other observations in $C \in \rho$, weighted by the uncertainty in $\rho$ as quantified by its posterior distribution conditional on $\rho \in \mathcal{O}(r_{0},\boldsymbol{\pi}_{1:n_{0}})$.

It is straightforward to sample from the distribution with density \eqref{eq:prob_theta_j_0} using \eqref{eq:prob_theta_j_3}. First sample a partition $\rho$ from its posterior distribution conditional on $\rho \in \mathcal{O}(r_0, \boldsymbol{\pi}_{1:n_0})$. Given this sampled partition, identify the cluster $C$ that contains the index set $C_{0,j}$, and then draw $\theta_j$ from its corresponding conditional posterior distribution based on the observations in that cluster.  There is no ambiguity in the choice of $C$ within $\rho$, since $\rho$ is assumed to be in $\mathcal{O}(r_0, \boldsymbol{\pi}_{1:n_0})$, which means that there is a unique cluster in $\rho$ that contains the observations in $C_{0,j}$. Since our approach is based on $r_{0,\bm{\pi}_{1:n_0}}$, the number of cluster-specific parameters for which inference can be made corresponds to the number of clusters in $r_{0,\bm{\pi}_{1:n_0}}$, that is, $k_0$. In this context, $k_0$ can be thought of as a lower bound estimate of the number of clusters.

Equation \eqref{eq:prob_theta_j_0} can be used to construct a credible interval for $\theta_j$ as follows.  Let $I_j \subset \Theta_j$ denote an interval or region of $\Theta$ that contains a given posterior probability, denoted as $\alpha$. Then a posterior joint credible region for $\rho$ and $\theta_j$ is based on the following posterior probability 
\begin{align*}
\Pr(\rho \in \mathcal{O}(r_0, \boldsymbol{\pi}_{1:n_0}), \theta_j \in I_j ~|~ \bm{y})   = 
\Pr(\theta_j \in I_j ~| ~ \rho \in \mathcal{O}(r_0, \boldsymbol{\pi}_{1:n_0}), ~ \bm{y})  
\Pr(\rho \in \mathcal{O}(r_0, \boldsymbol{\pi}_{1:n_0}) ~|~ \bm{y}) \geq  \alpha\gamma.  
\end{align*}
The posterior probability of the joint statement $[\rho \in \mathcal{O}(r_0, \boldsymbol{\pi}_{1:n_0}), \theta_j \in I_j]$ is upper-bounded by $\gamma > \Pr(\rho \in \mathcal{O}(r_0, \boldsymbol{\pi}_{1:n_0}) ~|~ \bm{y})$, 
while the posterior probability of region $I_j$
is lower bounded by $\alpha\gamma$. In practice, $\alpha$ and $\gamma$ should be set so that 
$\alpha\gamma$ equals the target credible level. To estimate a credible set for $\theta_j$ from MCMC samples, simply use the iterates for which $r_{0,\bm{\pi}_{1:n_0}}$ holds.  This approach is demonstrated in Section \ref{sec:illustrationssimulations}.

\subsection{Cluster-Level Posterior Probabilities}\label{sec:cluster_level_probs}
Given the subpartition 
$r_{0,\bm{\pi}_{1:n_0}} = \{C_{0,1}, \ldots, C_{0,k_0}\}$, 
it is possible to compute the posterior probability  of each $C_{0,j}$. 
This quantity can be viewed as a marginal probability, or more generally, as a higher-order co-clustering probability. 
Specifically, for each cluster $C_{0,j}$ in the subpartition, we compute
$
\Pr\big(\rho \in \mathcal{C}_j \mid \bm{y}\big),
$
where $\mathcal{C}_j \subseteq \mathcal{P}$ is such that if $r \in \mathcal{C}_j$ then there exists a $C \in r$ with $C_{0,j} \subseteq C$. This probability corresponds to the posterior probability that the observations in $C_{0,j}$ are clustered together, which by definition is greater than or equal to $\Pr(\rho \in \mathcal{O}(r_0, \boldsymbol{\pi}_{1:n_0}) ~|~ \bm{y})$.  For example, the posterior probabilities of the clusters in the subpartition displayed in the bottom left plot of Figure \ref{fig:toyexample} are respectively 0.911 - red, 0.795 - green, 0.786 - blue, and 0.919 - cyan.  


\subsection{Partition Estimation with CHIPS and SALSO}
It is also possible to produce a partition point estimate using the subpartition $r_{0,\bm{\pi}_{1:n_0}}$.  Once $r_{0,\bm{\pi}_{1:n_0}}$ is obtained given the threshold $\gamma$, the remaining $n - n_0$ units can be assigned to clusters using the SALSO algorithm (\citealt{dahl_etal:2022}). In the simulations, we show that  combining the CHIPS and SALSO algorithms can, under certain circumstances, produce a more accurate point estimate of $\rho$. 

\section{Simulations and Illustrations} \label{sec:illustrationssimulations}
Note that the previous section was written agnostic to the model used to produce MCMC samples of cluster labels.  This was intentional, as  $\mathcal{O}(r_0, \boldsymbol{\pi}_{1:n_0})$ can be computed regardless of the model used, so long as MCMC samples from the posterior distribution of $\rho$ are available. (For examples of Bayesian models with clustering, see \citealt{argiento&deIorio:2022} and \citealt{Mueller2015}.)  To illustrate the methodology of the previous sections, we will use the following finite mixture model.  Let $\bm{y}_i \in \mathbb{R}^d$ denote the $d$-dimensional response of the $i$th unit and $z_i$ denote its corresponding component label:  
\begin{align} \label{eq:finite_mixture}
\begin{split}
    \bm{y}_i ~|~ z_i & \stackrel{ind}{\sim} N(\bm{\mu}_{z_i}, \bm{\Sigma}_{z_i}), \ i = 1, \ldots, n\\
    \Pr(z_i = k ~|~ \bm{\pi}) & = \pi_k, \ k = 1, \ldots, M   
\end{split}    
\end{align}
In the simulations that we detail shortly, we consider two versions of \eqref{eq:finite_mixture}.  The first, which we denote as the ``z'' model, treats $z_i$ as the only unknown in the model.  That is, the atoms ($\bm{\mu}_k, \bm{\Sigma}_k$), weights $\bm{\pi}$, and number of components $M$ are all set to their true values used to generate the data.  For this model, all the partition uncertainty originates from the cluster labels and as a result permits visualizing more clearly how $\mathcal{O}(r_0, \boldsymbol{\pi}_{1:n_0})$ assesses uncertainty.    The second model, which we refer to as the ``z \& Atoms \& M'' model treats all components of the model as unknown.  This model is fit using custom R code based on the approach found in \cite{fruhwirth2021generalized} and the following prior distributions
$(\bm{\mu}_k, \bm{\Sigma}_k) \stackrel{iid}{\sim} {\rm MVNIW}(\bm{\mu}_0, \kappa_0, \nu_0, \bm{\Psi}_0)$ with $\bm{\mu}_0$ a $d$-dimensional vector of zeros, $\kappa_0 = 0.5$, $\nu_0 = d + 2$, $\bm{\Psi}_0 = 0.75\times {\rm diag}(d)$, and $\bm{\pi} \sim {\rm Dirichlet}(\alpha/K, \ldots, \alpha/K)$, with $\alpha \sim {\rm Gamma}(1,1)$, and $M-1 \sim {\rm Poisson}(4) $.  

\subsection{Synthetic Example Revisited}\label{sec:toy_example}
Consider the example introduced in Section \ref{sec:intuition_building}.  We sampled cluster labels $z_i$ based on the ``z''-model using their full conditionals.  This ensured that the only variability is due to the cluster labels and also that the point at $(0,0)$ had equal probability of being allocated to each of the clusters.   The results of running the CHIPS algorithm on the cluster labels are provided in Figure \ref{fig:toyexample}, where the colored points in the figure depict $r_{0,\boldsymbol{\pi}_{1:n_0}}$ and the gray points correspond to observations that do not belong to $r_{0,\boldsymbol{\pi}_{1:n_0}}$. The size of the gray points is proportional to their $p^{\rm max}_{n_0} - q^{\rm max}$ value.  Thus, adding points to $r_{0,\boldsymbol{\pi}_{1:n_0}}$  whose circumference is large reduces $\Pr(\rho \in \mathcal{O}(r_0, \boldsymbol{\pi}_{1:n_0}) ~|~ \bm{y})$ more drastically.  The ellipses in each figure correspond to  95\% credible regions computed using the approach detailed in Section \ref{sec:cluster_specific_inference} and the ``$\times$''  symbol corresponds to the true value of the atom used to generate data. 

The first column of Figure \ref{fig:toyexample} corresponds to running the CHIPS algorithm until an $r_0$ and $\boldsymbol{\pi}_{1:n_0}$ are found that produces $\Pr(\rho \in \mathcal{O}(r_0, \boldsymbol{\pi}_{1:n_0}) \mid \bm{y}) \ge 0.5$.  The second column is similar in that it displays the $r_0$ and $\boldsymbol{\pi}_{1:n_0}$ that produces $\Pr(\rho \in \mathcal{O}(r_0, \boldsymbol{\pi}_{1:n_0}) \mid \bm{y}) \ge 0.95$, and the third column displays the curves $(\ell, p^{\rm max}_{\ell})$, which characterize how the posterior probability of subpartitions behaves as the number of clustered items increases.  In these plots, we highlight the subpartitions that have 0.5 and 0.95 posterior probability using a horizontal line.   Each row of the figure corresponds to a different within cluster variance $\sigma^2$.  

Focusing on the top left plot of Figure \ref{fig:toyexample}, the subpartition that has posterior probability of at least 0.5 includes all points except the center point.  Including this point in the partition decreases the posterior probability of the partition to 0.25 as expected and as seen in the $(\ell, p^{\rm max}_{\ell})$ curve plot, where the last point allocated reduces the partition probability to 0.25.   The subpartition with at least 0.95 posterior probability displayed in the top middle figure excludes two additional points, but the middle point still reduces the subpartition probability to about 0.25. As cluster overlap increases, the subpartitions contain fewer units; the subpartition with a posterior probability of at least 0.95 with the largest overlap is made up of only 14 of the $n=100$ observations.  Even with so few observations, using $r_{0,\boldsymbol{\pi}_{1:n_0}}$ as described in Section \ref{sec:cluster_specific_inference} to produce credible intervals for cluster-specific parameters results in intervals that almost contain all of the values used to generate the data.

The last column of Figure \ref{fig:toyexample} presents the curves $(\ell, p^{\rm max}_{\ell})$, which characterize how the posterior probability of subpartitions behaves as the number of clustered items increases. In the last panel of the first row, where the within-cluster variance is low ($\sigma^2 = 0.1$), the curve remains nearly constant and close to one, indicating low uncertainty, as the posterior distribution concentrates most of its mass on a small number of partitions, allowing a moderate or large number of individuals to be clustered with high posterior probability. As $\sigma^2$ increases, the overlap between the distributions associated with each cluster grows, leading to more uncertainty in the partitions. This is reflected in the last panels of the second and third rows, where the curves decline more rapidly, showing that the posterior probability is more evenly spread across different, less distinct partitions, making it difficult to cluster a moderate number of individuals with high posterior probability. 

In the first row of Figure \ref{fig:toyexample}, where $\sigma^2 = 0.1$, the AUChips is approximately 0.996, confirming that the posterior distribution is highly concentrated on a small number of partitions. In the second row, with $\sigma^2 = 0.25$, the AUChips drops to 0.803, reflecting increased uncertainty. In the third row, where $\sigma^2 = 0.5$, the AUChips further decreases to 0.405, indicating that the posterior probability is more spread across many different partitions, making it possible to cluster only a small number of individuals with high posterior probability.

\subsection{Simulation Study}
\label{sec:simulation_study}
We now conduct a simulation study to investigate several aspects of our proposed methodology.  To do this we generate datasets similar to that displayed in Figure \ref{fig:toyexample}, but with a more granular sequence of within cluster variance $\sigma \in \{0.3, 0.4, 0.5, 0.6, 0.7, 0.8, 0.9, 1\}$.  Two sample sizes are considered, one with 25 observations per cluster for a total of $n=100$ observations and the other with 100 observations per cluster for a total of $n=400$ observations. One hundred datasets were generated and to each we fit the ``z" and ``z \& Atoms \& M'' models.  For each model fit, we found $\mathcal{O}(r_0, \boldsymbol{\pi}_{1:n_0})$ for $\gamma \in \{0.5, 0.75, 0.95\}$ using both the variation of information (VI) and Binder loss functions.  Results for $n=100$ and for VI loss are provided in Figures \ref{fig:simulation_compare_with_salso} and \ref{fig:simulation_compare_with_wade}.  In the online supplementary material, we provide results for $n=400$ under VI and for both sample sizes under Binder loss.  As expected, simulation results are quite dependent on the loss function employed.  In some instances, relationships highlighted here remain the same but, in others, they change relative to Binder loss.  See the supplementary material for more details.

We first focus on results from the ``z'' model fit shown in the first column of Figure \ref{fig:simulation_compare_with_salso}.  The first row  displays the Rand index (RI) computed for only the observations that are contained in $\mathcal{O}(r_0, \boldsymbol{\pi}_{1:n_0})$.  As expected the RI is very close to 1 and remains relatively constant, except when $\gamma = 0.5$, in which it decreases as a function of cluster overlap.  The reason this is expected is because the cardinality of $\mathcal{O}(r_0, \boldsymbol{\pi}_{1:n_0})$ decreases monotonically as a function of cluster overlap, also shown in the third row of Figure \ref{fig:simulation_compare_with_salso}.  This same pattern holds for AUChips (see fourth row).  As the cluster overlap increases, uncertainty in the cluster labels increases which results in a smaller AUChips.  The second row displays the difference in the RI between the CHIPS \& SALSO point estimate and the SALSO point estimate. Interestingly, it appears when overlap increases, the CHIPS \& SALSO point estimate performs better under the ``z'' model. However, this not the case when the ``z \& Atoms \& M'' model is fit.  The take-home message for the ``z'' model is that, as cluster overlap increases and for given threshold $\gamma$, AUChips decreases and the subpartition contains less items.  Focusing now on the second column which corresponds to the ``z \& Atoms \& M'' (a more natural model fit in practice), the narrative changes slightly in that when there is massive overlap, it is often the case there is small uncertainty associated with partition (see AUChips) but the partition is wrong (see the RI subpartition).  This is actually to be expected because the MCMC samples of $z_i$ with large cluster overlap allocate all units to one just cluster with high probability.

\begin{figure}
\begin{center}
\includegraphics[scale=0.55, page=1]{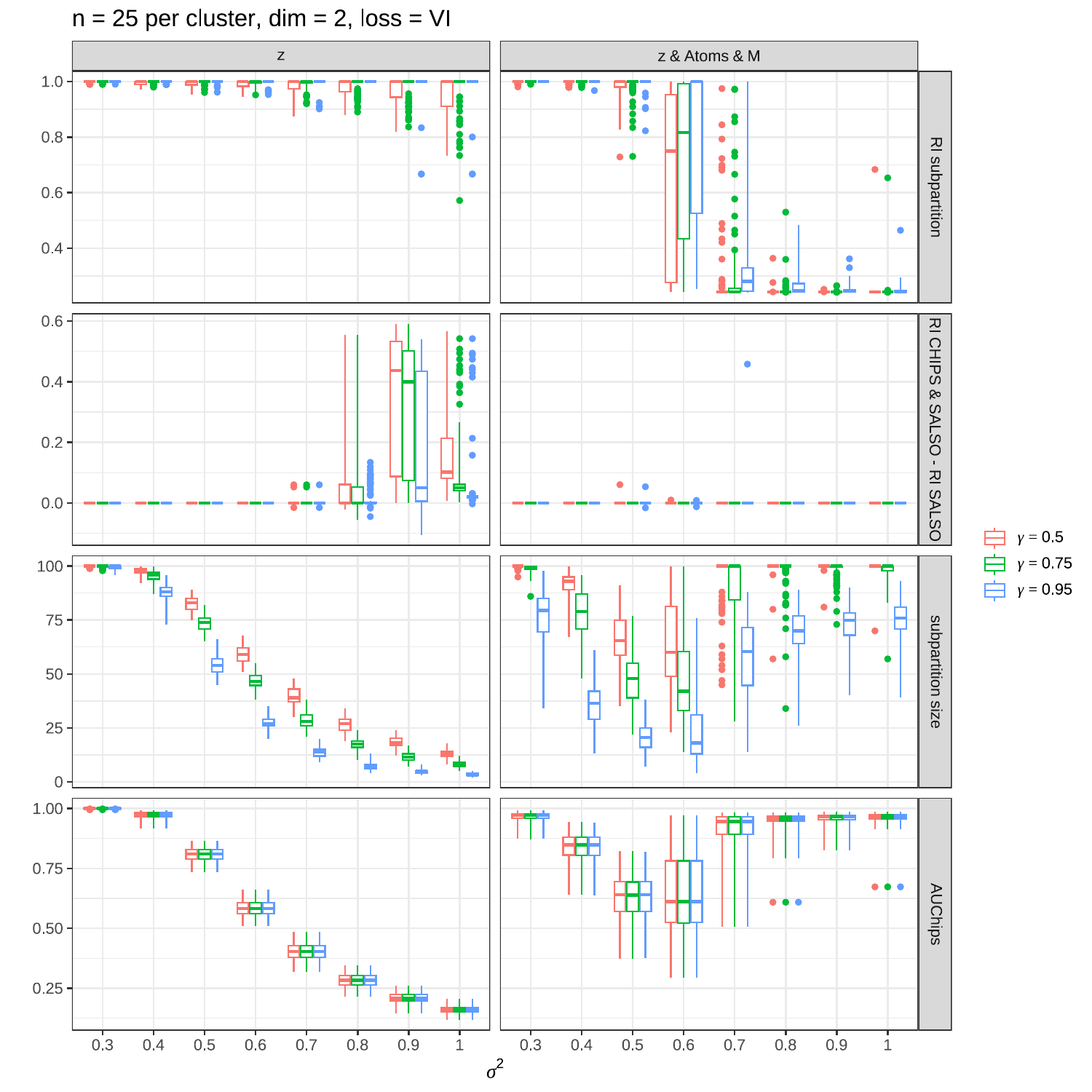}
\end{center}
\caption{Results from the simulation study in Section \ref{sec:simulation_study}. The third row displays the number of observations included in the subpartition and the bottom row shows the AUChips.}
\label{fig:simulation_compare_with_salso}
\end{figure}


In Figure \ref{fig:simulation_compare_with_wade}, we compare our $\mathcal{O}(r_0, \boldsymbol{\pi}_{1:n_0})$ to \cite{wade&ghahramani:2018}'s $\mathcal{B}_{\epsilon}(r_0)$.  The RI between partitions in $\mathcal{B}_{\epsilon}(r_0)$ and the true partition is on average higher than for the partitions in $\mathcal{O}(r_0, \boldsymbol{\pi}_{1:n_0})$.  This is somewhat expected since $\mathcal{O}(r_0, \boldsymbol{\pi}_{1:n_0})$ contains all $r \in \mathcal{P}$ for which $r_{0,\bm{\pi}_{1:\ell}}$ is a subpartition.  Generally speaking, this results in $|\mathcal{O}(r_0, \boldsymbol{\pi}_{1:n_0})| > |\mathcal{B}_{\epsilon}(r_0)|$ and therefore results in more partitions that have a small RI with the true partition. However,  when computing RI only for elements in $r_{0,\bm{\pi}_{1:\ell}}$ --- shown in the second row of Figure \ref{fig:simulation_compare_with_wade} --- the RI between the partitions in $\mathcal{B}_{\epsilon}(r_0)$ and the truth are worse on average when model ``z'' is used and the RI value decays as $\sigma$ increases.  That said, the differences are much smaller under the ``z \& Atoms \& M'' model.    The variability of RI between partitions in $\mathcal{B}_{\epsilon}(r_0)$ and the true partition is slightly higher than that for those in $\mathcal{O}(r_0, \boldsymbol{\pi}_{1:n_0})$ and this is particularly true when only considering elements that belong to $r_{0,\bm{\pi}_{1:\ell}}$.  The upshot is that our method performs (in terms of Rand index to the truth) about as well as does the method of \cite{wade&ghahramani:2018}, yet our method provides a more interpretable understanding of the how partitions in the credible set are related to each other.

\begin{figure}
\begin{center}
\includegraphics[scale=0.55, page=2]{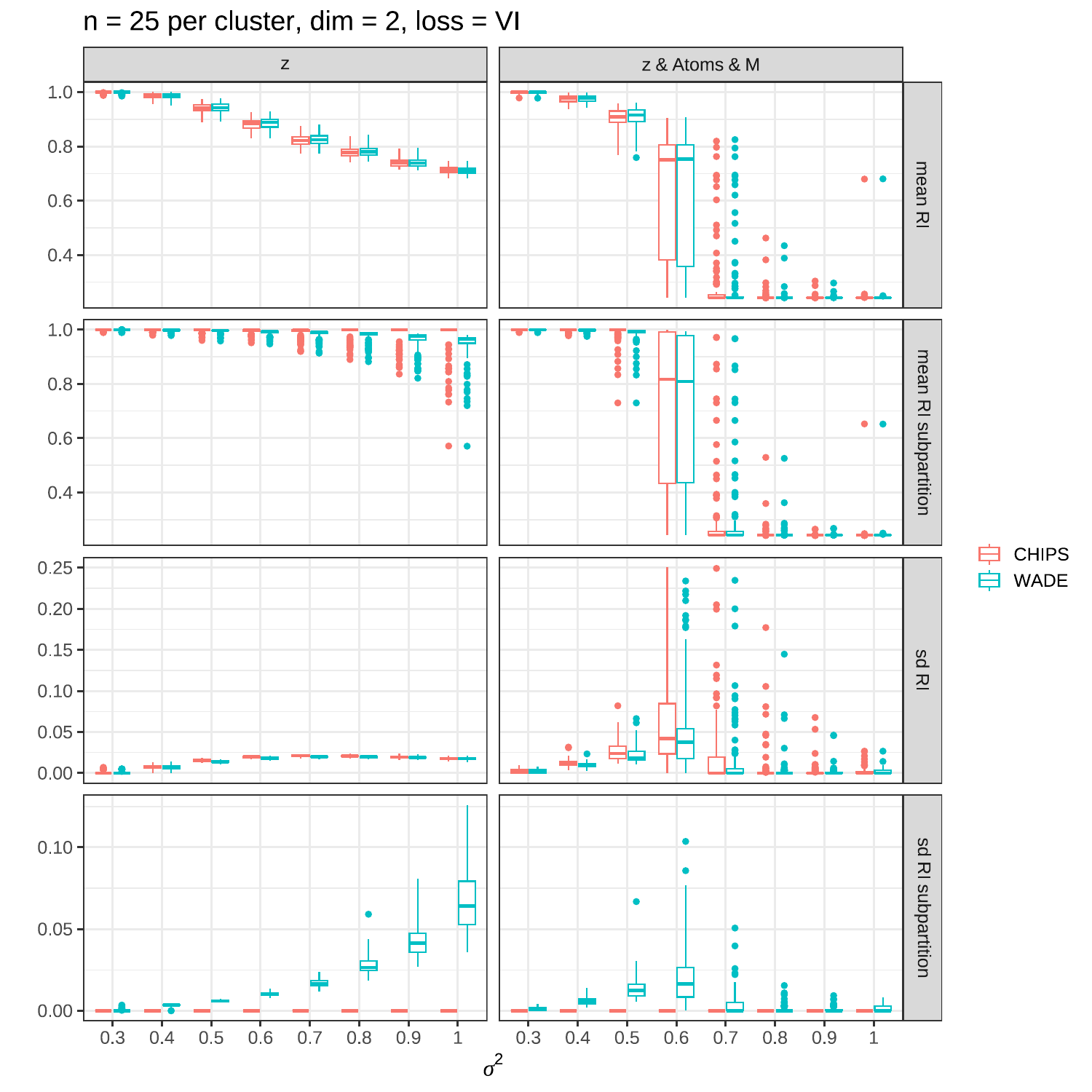}
\end{center}
\caption{Results from the simulation study comparing properties of our interpretable credible set to that of the credible ball of \cite{wade&ghahramani:2018}.  The threshold was set to $\gamma = 0.75$ when forming both credible sets.}
\label{fig:simulation_compare_with_wade}
\end{figure}

Finally, Figure \ref{fig:cluster_parameter_estimation} shows the results for estimating cluster-specific parameters.  In both plots  ``CHIPS''  corresponds to the approach detailed in Section \ref{sec:cluster_specific_inference} and ``SALSO'' corresponds to conditioning on the salso partition estimate.  The plot on the left corresponds to the average number of detected clusters, and the plot on the right displays the coverage percentage associated with the detected clusters. Results for ``SALSO'' do not change as a function of $\gamma$.  As expected, as cluster overlap increases (i.e., $\sigma^2$ increases), detecting the number of clusters becomes more challenging.  Under the ``z'' model, our approach is more accurate at estimating the number of clusters for small $\gamma$ values.  For the ``z \& Atoms \& M'' model, our approach underestimates the number of clusters relative to ``SALSO''.   
However, our approach maintains good coverage for the ``z'' model regardless of cluster overlap, while  coverage decreases under the ``z \& Atoms \& M'' model as cluster overlap increases. This suggests that reduced coverage reflects model performance rather than a limitation of our method for constructing suitable credible regions. Across all scenarios, our approach maintains a coverage rate that is at least as good as the approach that conditions on the SALSO estimate.
The results displayed in Figure \ref{fig:cluster_parameter_estimation} are based on VI loss; those based on Binder loss result in somewhat different patterns and are provided in the supplementary material.


\begin{figure}[htbp]
  \centering
  \begin{minipage}{0.44\textwidth}
    \centering
    \includegraphics[width=\linewidth, page=3]{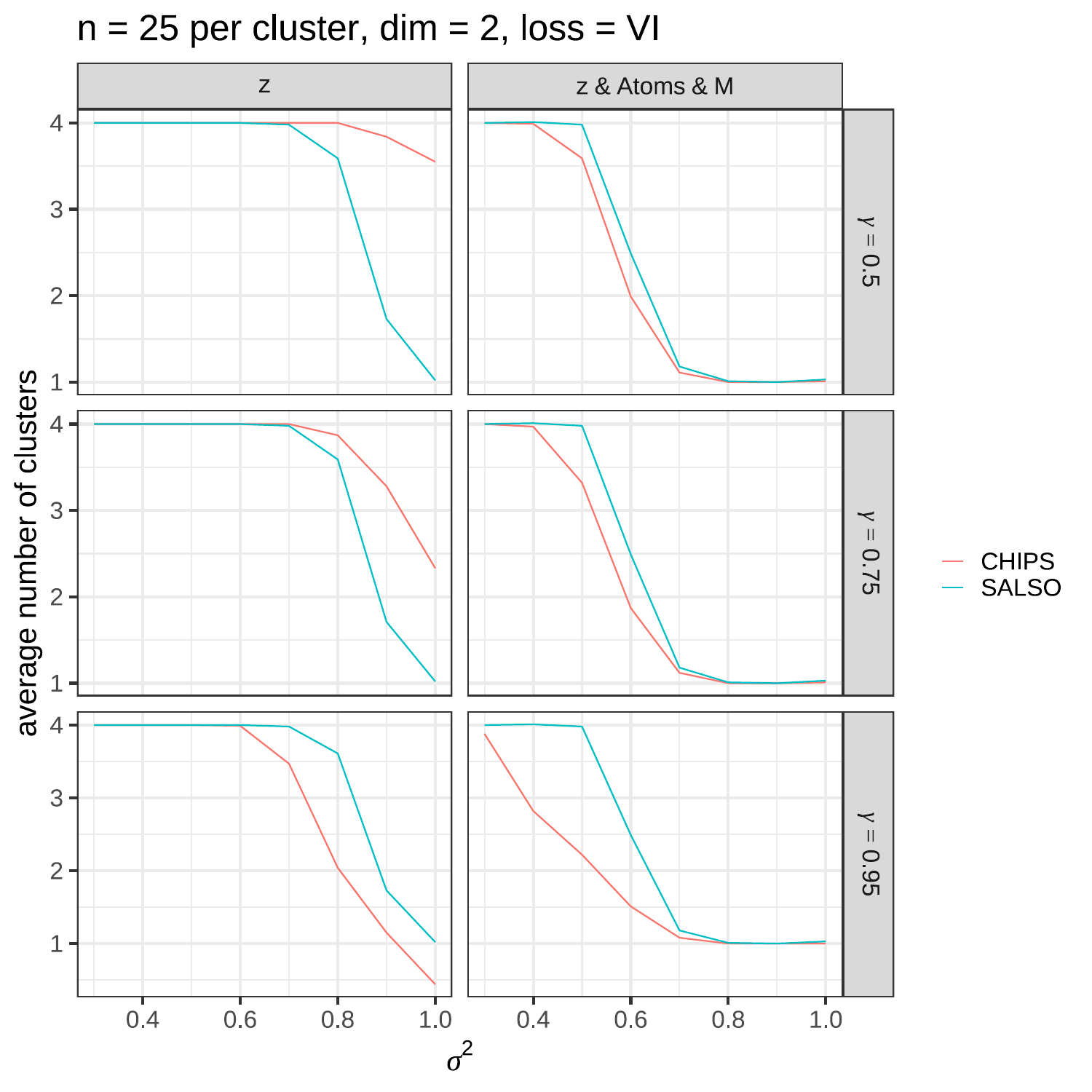}
  \end{minipage}
  \hfill
  \begin{minipage}{0.44\textwidth}
    \centering
    \includegraphics[width=\linewidth, page=4]{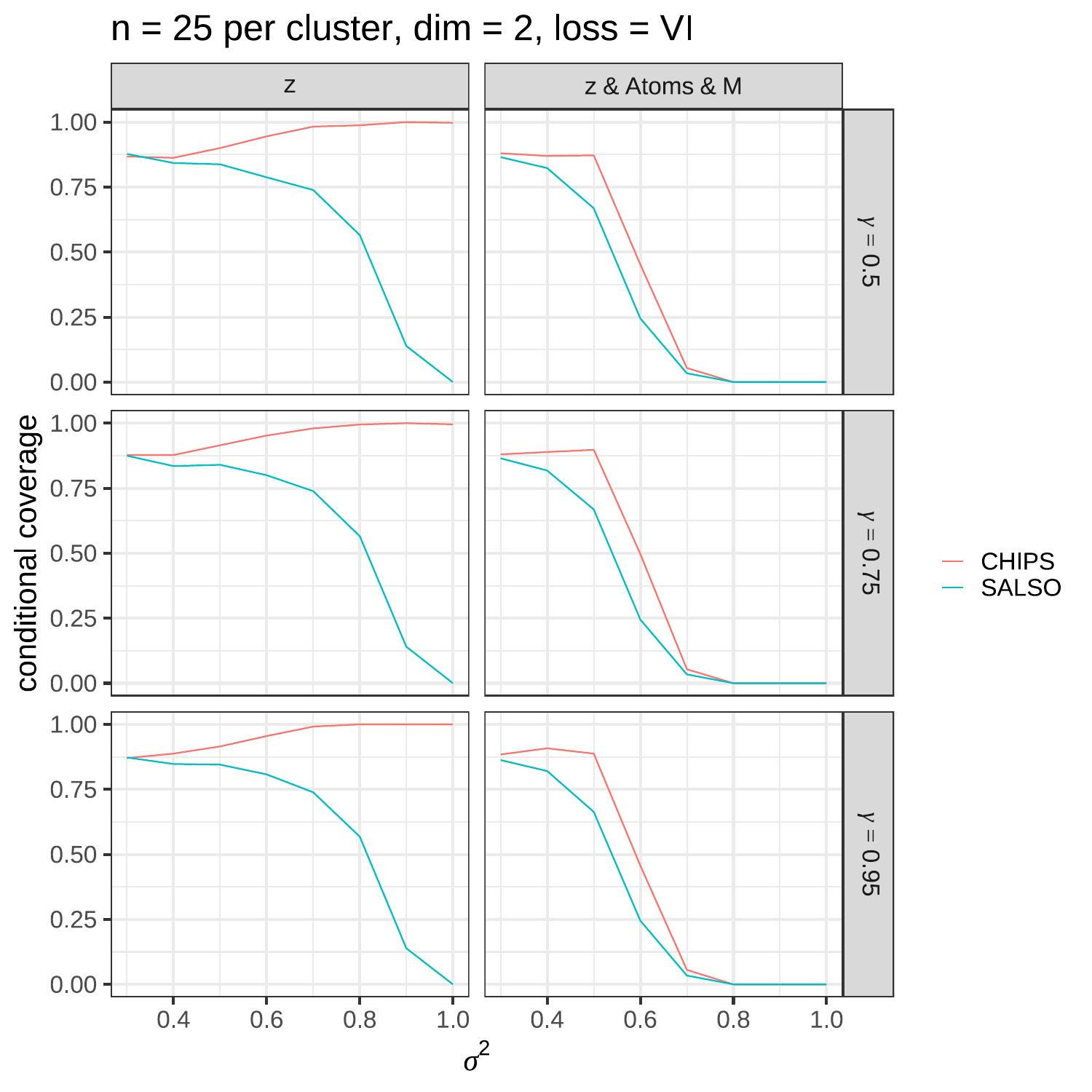}
  \end{minipage}
  \caption{Results from simulation study that compares the approach of producing credible intervals detailed in Section \ref{sec:cluster_specific_inference}, and conditioning on the SALSO point estimate.  The left figure displays results associated with the number of clusters detected and the right with the number of detected clusters that resulted in credible regions that covered the truth.} \label{fig:cluster_parameter_estimation}
\end{figure}



\subsection{Exercise Science Application}

Our first nonsynthetic-data demonstration comes from the field of biomechanics and is a functional data application in which inference about cluster-specific parameters is of interest in addition to uncertainty quantification.  To provide context, biomechanics is the study of mechanical principles (e.g., forces and angles) applied to living organisms. A common focus of biomechanical studies is to discover how forces and angles affect musculoskeletal health.  Recently researchers have learned that anterior cruciate ligament injury and the subsequent reconstructive surgery (ACLR) affect walking biomechanics, in different ways, for different ACLR patients.   For example, some ACLR patients walk with a knee that is flexed less during the early part of the ground contact phase of walking, which is often associated with decreased vertical ground reaction force (vGRF) during the same walking phase. Further, these ACLR patients, who experience less vGRF during the early part of ground contact, often report less desirable outcomes in areas such as pain, stiffness, or quality of life (\citealt{hayden_etal:2024}). In an attempt to discover the different ways ACLR affects walking biomechanics,  196 ACLR patients were recruited to participate in a study that required them to walk on a treadmill during which vGRF was measured through the entire gait cycle.  The raw data are displayed in the top left plot of Figure \ref{fig:biomechanic_vgrf_results}. To these data, a variant of the model-based clustering approach described in \cite{page&quintana:2015} was fit using the {\tt ppmSuite} {\tt R}-package (\citealt{page_ppmSuite:2025}) by collecting 1000 MCMC samples from the posterior distribution of the cluster labels.  As a point of information the posterior distribution on the number of clusters was concentrated on two values, four clusters with 0.978 posterior probability and five clusters with 0.022  posterior probability.  

We applied CHIPS for $\gamma = 0.95$ to the collected MCMC samples.  The results are shown in Figure \ref{fig:biomechanic_vgrf_results}.  The top right panel displays the CHIPS \& SALSO point estimate (the same as the SALSO point estimate in this case), which has a posterior probability close to $0$, and the bottom left plot displays the subpartition consisting of 53 subjects, which has a posterior probability of 0.953.  The AUChips value is 0.51, indicating substantial uncertainty.  The left plot of Figure \ref{fig:biomechanic_cluster_specific_parm_colclustering} displays the pairwise co-clustering probability matrix. Although there is a high posterior probability of four clusters, there is uncertainty associated with which units are clustered together.  The four clusters in the subpartition have posterior probabilities 0.993 - black, 0.983 - red, 0.981 - green, and 0.99 - blue (see Section \ref{sec:cluster_level_probs}).  The dashed gray line in the bottom plot of Figure \ref{fig:biomechanic_vgrf_results} corresponds to the subject with the smallest $q^{\rm max}$ value.  Including this subject in the subpartition would reduce the posterior probability from 0.95 to 0.41. This is consistent with the unusual curve for this subject, with a deep valley similar to the red cluster but with peaks distinct from the other curves.

In this application, inference associated with cluster-specific mean vGRF curves is of interest.  Applying the approach detailed in Section \ref{sec:cluster_specific_inference} results in the right plot of Figure \ref{fig:biomechanic_cluster_specific_parm_colclustering}.  The figure displays 95\% credible bounds for each cluster's mean curve.  As a result, the level of the joint credible region for partition and curve is $\alpha\times \Pr(\rho \in \mathcal{O}(r_0, \boldsymbol{\pi}_{1:n_0}) \mid \bm{y}) = 0.95 \times 0.953 = 0.90535$.  The red cluster is the type of vGRF curve that biomechanists attach to healthy individuals indicating that this group of ACLR patients have not altered biomechanics, while the purple cluster curve displays large changes in walking biomechanics.

\begin{figure}
\begin{center}
\includegraphics[scale=0.7]{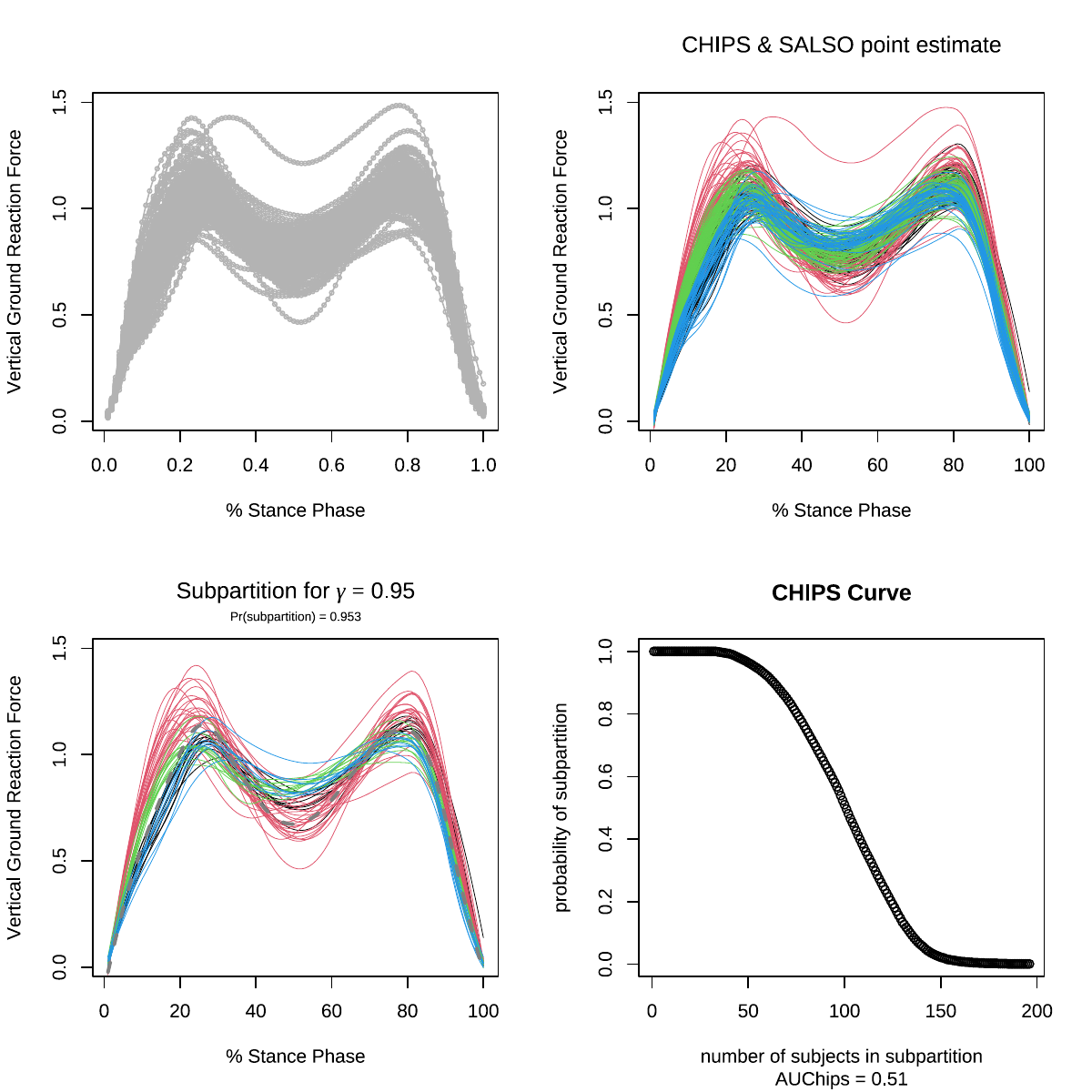}
\end{center}
\caption{The top left figure displays the vGRF curve for each of the 196 subjects, the top right gives a point estimate of $\rho$ using CHIPS \& SALSO, the bottom left displays that subpartition which has posterior probability of 0.95, and the bottom right is the CHIPS curve.}
\label{fig:biomechanic_vgrf_results}
\end{figure}



\begin{figure}[htbp]  
  \centering
  \begin{minipage}{0.44\textwidth} 
    \centering
    \includegraphics[width=\linewidth]{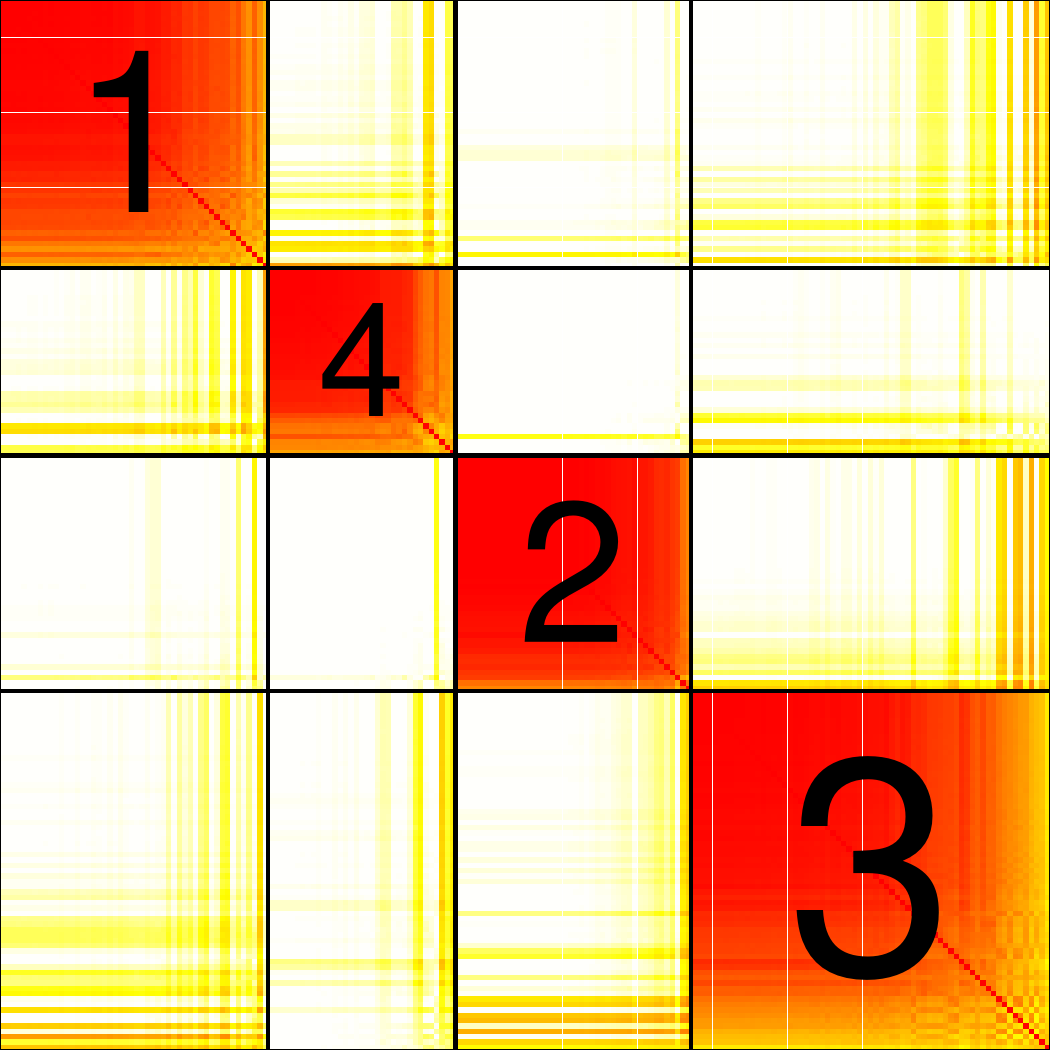}
  \end{minipage}
  \hfill
  \begin{minipage}{0.44\textwidth}
    \centering
    \includegraphics[width=\linewidth]{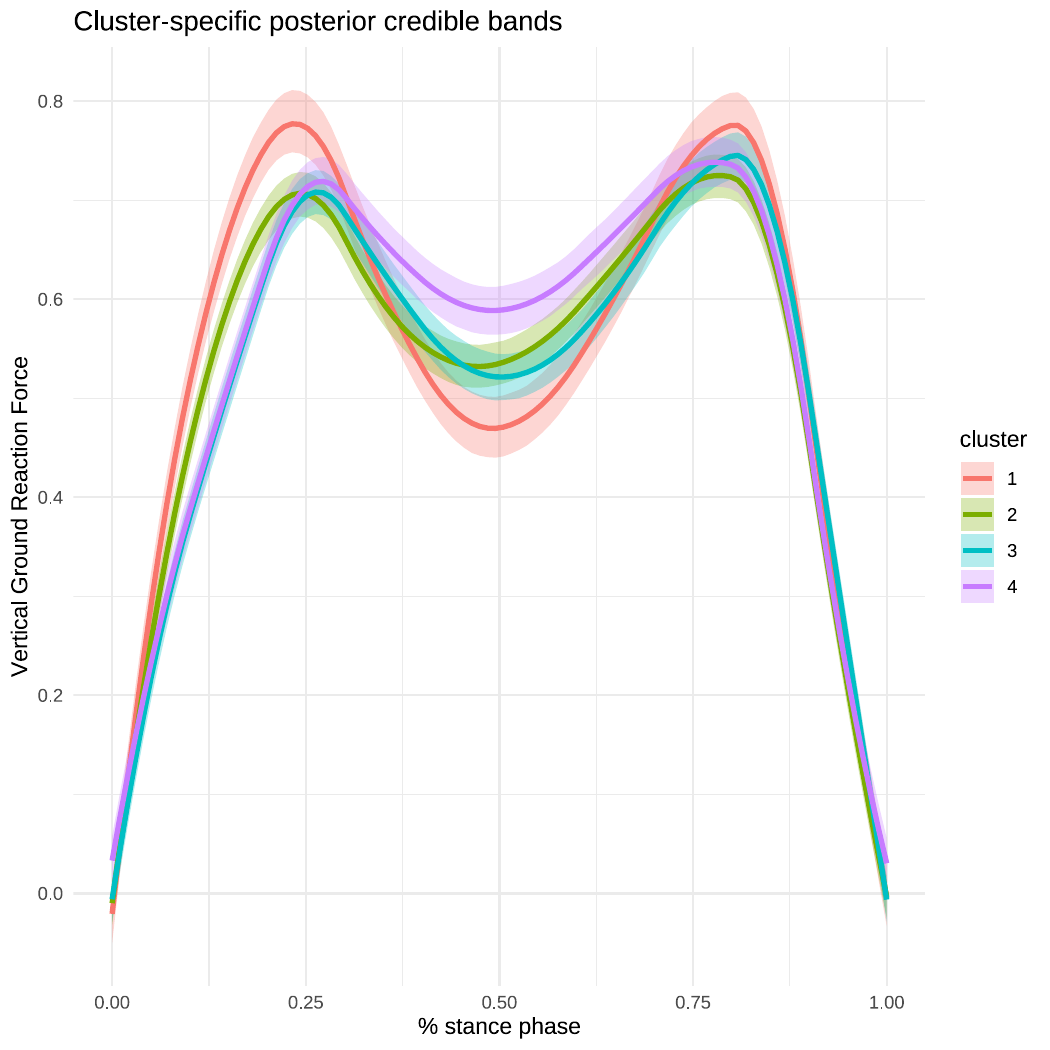}
  \end{minipage}
  \caption{The left plot displays the pairwise co-clustering probability matrix.  The right plot displays the cluster-specific estimated mean curves and 95\% credible bands.}
  \label{fig:biomechanic_cluster_specific_parm_colclustering}
\end{figure}

\subsection{Authorship of the Federalist Papers}
Our final example illustrates a situation in which the subpartition itself can be of interest beyond the uncertainty quantification that it provides.  The Federalist Papers are a collection of 85 essays that were very influential in early USA history with regards to ratifying its constitution.  While the authorship of 74 of the essays is well established, there are 11 whose authorship is still disputed.  Attempts to determine authorship of the remaining 11 have appeared in the literature (\citealt{mosteller:1963, jeong2025smalllargelanguagemodels}).  We attempt to shed further light on the subject using our CHIPS approach.  

The data consist of 85 essays and 70 functional words  that are normalized by the total number of words per essay\footnote{https://www.kaggle.com/datasets/tobyanderson/federalist-papers}.  Our goal is to cluster the 85 essays using the 70 dimensional vector of functional words with the idea that documents with similar word distributions were written by the same author.  Due to the $n \approx p$ setting,  we employ the factor model approach of \cite{chandra_etal:2023} which avoids the curse of dimensionality in clustering.  Specific details associated with the prior values employed and other important model details are provided in the supplementary material.  Here we simply mention that the posterior distribution associated with the partition is sensitive to these model decisions and so we emphasize the fact that the results that follow are dependent on the specified model.

After collecting 1000 MCMC samples from the posterior distribution of $\rho$ based on the model details provided in the supplementary material, we apply the CHIPS algorithm with $\gamma = 0.95$.  This resulted in AUChips = 0.90 which is fairly high, and so the posterior distribution of $\rho$ is fairly concentrated.  The CHIPS \& SALSO point estimate produced five clusters whose sizes are 75, 5, 3, 1, 1.  The subpartition included 55 of the 85 essays and  consisted of five clusters whose sizes are 45, 5, 3, 1, 1 all with posterior probability 1 except for the large cluster which has posterior probability of 0.96.  The smaller clusters in the subpartition were all pure in the sense that cluster 2 contained 5 documents known to be authored by John Jay, cluster 3 contained 3 documents authored jointly by Hamilton and Madison, and clusters 4 and 5 are singletons of Hamilton authored documents. No Madison authored essay belonged to the subpartition.  Interestingly the large cluster of 45 essays consisted of two disputed essays (55 and 62) and 43 known Hamilton essays. This runs counter to the analysis \cite{mosteller:1963} who conclude that  all the disputed essays are authored by Madison.  That said, they do mention that evidence in favor of essay 55 being authored by Madison is somewhat weak, and historians generally believe that essay 62 was authored by Hamilton, not Madison.   Of the Madison essays, 11 of them had the largest $q^{\max}$. In fact, adding any of the 11 essays to the subpartition would reduce the posterior probability of the subpartition from 0.9682 to 0.6868 which turns out to be the posterior probability of the SALSO point estimate.

\section{Conclusions}\label{sec:conclusion}

In Bayesian inferences, the vast majority of the literature places all of the effort in the posterior formulation, including the choice of prior and likelihood and how to conduct posterior computation. Critical questions of how best to report results of the analysis are often left unanswered, and instead practitioners use simple posterior summaries. Although this may be sufficient in many cases, for complex models and challenging parameter spaces, questions of how to best summarize the posterior are non-trivial and greatly impact the main conclusions of statistical analyses.

This paper establishes an interpretable framework for carrying out uncertainty quantification in Bayesian clustering. Our approach is based on constructing a sequence of subpartitions using the CHIPS algorithm.  The algorithm produces a credible set that accumulates a prespecified amount of posterior mass that is very interpretable in the sense that all partitions in the subset share a common subpartition.   The approach produces an overall measure of partition posterior uncertainty, the first of its kind, as far as we are aware. In addition, coherent local measures of uncertainty are provided for all units that do not belong to to the subpartition.    We demonstrated how the subpartition can be employed to carry out cluster-specific parameter inference without having to use relabeling techniques or conditioning on a partition point estimate that oftentimes has low posterior probability.  All of this is available in general software that does not rely on specifics of any given model.  Collectively, these contributions greatly enhance the interpretability of Bayesian clustering by establishing a foundation for principled uncertainty quantification for partition-based inference. Although we have focused on the case of summarizing posteriors on partitions, we hope that this work inspires more focus on improved posterior summaries in other challenging cases.

{\noindent {\bf Acknowledgments}} \\
This work was partially supported by grant R01ES035625 from the National Institute of Health (NIH), N000142412626 from the Office of Naval Research (ONR) and IIS-2426762 from the National Science Foundation (NSF).  

\singlespacing
\bibliographystyle{asa}
\bibliography{refs}

\appendix
\doublespacing


\newpage
\section*{Appendix}  

\section{CHIPS: Forward Algorithm}\label{sec:chips}

The algorithm below takes as input the MCMC draws $\rho^{(1)}, \ldots, \rho^{(M)}$, which are used to approximate the posterior probabilities required in the procedure. The output of the algorithm is two-fold. First, it returns the credible set $\mathcal{O}(r_0, \boldsymbol{\pi}_{1:n_0})$ that accumulates at least $\gamma$ of the posterior probability. Second, it produces the set
$\{(r^{(i_1)}, \boldsymbol{\pi}^{(i_1)})\}_{i_1=1}^n$, 
where each element characterizes a monotone sequence of subsets
$
\{r^{(i_1)}\} = \mathcal{O}\left(r^{(i_1)}, \boldsymbol{\pi}^{(i_1)}_{1:n}\right) \subset
\ldots \subset
\mathcal{O}\left(r^{(i_1)}, \boldsymbol{\pi}^{(i_1)}_{1:1}\right) = \mathcal{P}
$ constructed through a greedy strategy. These monotone sequences serve as the input for computing the AUChips measure described in the manuscript. 

Let $\mathcal{I} = (i_1^{(1)}, \ldots, i_1^{(M)})$ denote a list of starting points, where each element of the list takes values in $\{1,\ldots,n\}$. This allows the same starting point to appear multiple times in $\mathcal{I}$. For notational convenience, we will write $i_1 \in \mathcal{I}$ when referring to elements of this list.

\subsection*{Algorithm}
\begin{enumerate}

\item  For each starting point $i_1$ of the list $\mathcal{I}$, do the following:
\begin{enumerate}
    \item[1.1.] Define $\tilde{\boldsymbol{\pi}}_{1:1} = \{i_1\}$, $k_1 = 1$, and $s_{1} = \{\tilde{C}_{1,1}\}$, where $\tilde{C}_{1,1} = \{i_1\}$.
    \item[1.2.] For $\ell = 2, \ldots, n$ find the following
    \begin{itemize}
        \item[1.2.1.]  Given $\tilde{\boldsymbol{\pi}}_{1:(\ell-1)} = \{i_1, \ldots, i_{\ell-1}\}$  and $s_{\ell-1} = \{\tilde{C}_{1,\ell-1}, \ldots, \tilde{C}_{k_{\ell-1},\ell-1}\}$, and for each of the remaining units $i \in  \left\{1, \ldots, n\right\}  \setminus \tilde{\boldsymbol{\pi}}_{1:(\ell-1)}$, define the following partitions of $\tilde{\boldsymbol{\pi}}_{1:(\ell-1)} \cup \{i\}$ based on $s_{\ell-1}$,
    \begin{align*}
     \tilde{s}_{1,i} & = \{\tilde{C}_{1,\ell-1}\cup \{i\}, \ldots, \tilde{C}_{k_{\ell-1},{\ell-1}} \} \\
     \vdots & \\
     \tilde{s}_{k_{\ell-1},i} & = \{\tilde{C}_{1,\ell-1}, \ldots, \tilde{C}_{k_{\ell-1},{\ell-1}}\cup \{i\} \} \\
    \tilde{s}_{k_{\ell-1}+1,i} & = \{\tilde{C}_{1,\ell-1}, \ldots, \tilde{C}_{k_{\ell-1},\ell-1}, \{i\}\}.
    \end{align*}
    \item[1.2.2.] Determine which observation should be added to $\boldsymbol{\pi}_{1:(\ell-1)}$ and whether it is added to one of the existing clusters in $s_{\ell-1}$ or forms a new one, that is, compute
        $$
       (i_{\ell}, j) = \argmax_{(i,t) \in ( \left\{1, \ldots, n\right\} \setminus \tilde{\boldsymbol{\pi}}_{1:(\ell-1)}) \times \{1, \ldots, k_{\ell-1}+1\} }Pr\left(\rho_{\tilde{\boldsymbol{\pi}}_{1:(\ell-1)} \cup \{i\}} = \tilde{s}_{t,i} ~ | ~ \bm{y}\right)
       $$
       {\bf Remark}: The solutions to the $\argmax$ statement may not be unique. In such cases, we randomly select one of these solutions.
       \item[1.2.3.] Define $\tilde{\boldsymbol{\pi}}_{1:\ell} = \{i_1, \ldots, i_\ell\}$ and  $s_{\ell} = \{\tilde{C}_{1,\ell}, \ldots, \tilde{C}_{k_{\ell},\ell}\}$ where
       \begin{itemize}
           \item[--] if $j \leq k_{\ell-1}$, then $k_{\ell} = k_{\ell-1}$,  $\tilde{C}_{j,\ell} = \tilde{C}_{j,\ell-1} \cup \{i_\ell\}$, and for $l  = 1, \ldots, j-1, j+1, \ldots, k_{\ell}$, $\tilde{C}_{l,\ell} = \tilde{C}_{l,\ell-1}$.
           \item[--] if $j = k_{\ell-1}+1$, then $k_{\ell} = k_{\ell-1}+1$, for $l  = 1, \ldots, k_{\ell}-1$, $\tilde{C}_{l,\ell} = \tilde{C}_{l,\ell-1}$, and $\tilde{C}_{l,k_{\ell}} = \{i_\ell\}$.
       \end{itemize}
    \end{itemize}
    \item[1.3.] Define $r^{(i_1)} = s_n$ and $\boldsymbol{\pi}^{(i_1)} = (i_1,\ldots,i_n)$ and use them to define the monotone sequence of subsets
$
\{r^{(i_1)}\} = \mathcal{O}\left(r^{(i_1)}, \boldsymbol{\pi}^{(i_1)}_{1:n}\right) \subset
\ldots \subset
\mathcal{O}\left(r^{(i_1)}, \boldsymbol{\pi}^{(i_1)}_{1:1}\right) = \mathcal{P}.
$
\end{enumerate}

     \item Using the monotone sequences obtained in the previous step, find the credible set by first finding
\begin{align*}
    n_0 = \max\left\{\ell : \Pr(\rho \in \mathcal{O}(r^{(i_1)},\boldsymbol{\pi}_{1:\ell}^{(i_1)} \mid \bm{y}) \ge \gamma, \, \mbox{or each starting point $i_1$ of the list $\mathcal{I}$} \right\},
\end{align*}
and then given $n_0$ finding      
\begin{align*}
(r_0,\boldsymbol{\pi}) =
\underset{i_1 \in \mathcal{I}}{\rm argmax} 
\Pr(\rho \in \mathcal{O}(r^{(i_1)},\boldsymbol{\pi}_{1:n_0}^{(i_1)}) \mid \bm{y}).
\end{align*}
and return $\mathcal{O}(r_0, \boldsymbol{\pi}_{1:n_0})$ and $\{(r^{(i_1)}, \boldsymbol{\pi}^{(i_1)})\}_{i=1}^n$.
\end{enumerate}

{\bf Remark}: 
Because the random selection that may occur at Step 1.2.2 can produce different outputs even when starting from the same observation, it would be ideal to run the algorithm multiple times for each of the $n$ possible starting points. However, this can be computationally demanding. When this is not feasible, we instead create $\mathcal{I}$ by randomly drawing starting points from $\{1,\ldots,n\}$ with replacement.

\subsection{Stability of the CHIPS Algorithm}
As mentioned, including every observation multiple times as a possible starting point would be ideal, since repeated use of the same $i_1$ can lead to different outputs due to the random selection that may occur at Step 1.2.2. However, creating such a large list of starting points is often computationally infeasible. When this is the case, a practical alternative is to create a more manageable list $\mathcal{I}$ by sampling from $\{1,\ldots,n\}$ with replacement and then running the algorithm once using all $i_1 \in \mathcal{I}$ as starting points. This approach reduces computational burden while still allowing us to assess the uncertainty captured by the posterior distribution of $\rho$. The exact size of $\mathcal{I}$ depends on the diffuseness of the posterior distribution, with more diffuse posteriors requiring larger lists. We explore this issue in the small simulation that follows. The {\tt chips} function in the {\tt salso} {\tt R}-package includes an argument {\tt nRuns} that allows the user to control the size of the list $\mathcal{I}$ used by the algorithm.
\begin{figure}
\begin{center}
\includegraphics[scale=0.65]{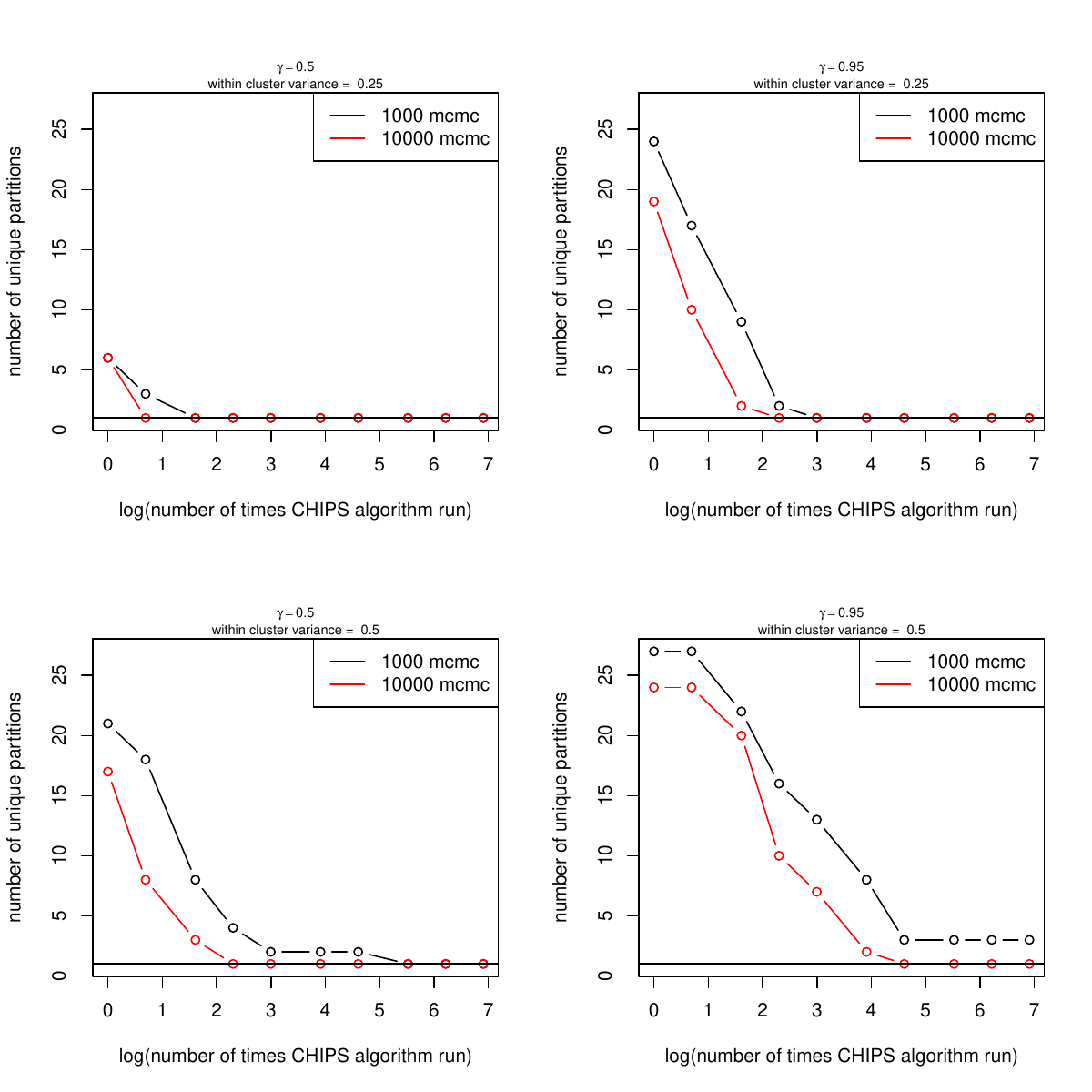}
\end{center}
\caption{Results of running CHIPS algorithm based on synthetic datasets in bottom two rows of Figure 1 found in the main document. The number of unique subpartitions out of 30 replicated experiments for different {\tt nRun} values are displayed }
\label{fig:stability1}
\end{figure}

Due to the discrete nature of MCMC approximation of the posterior probability of $\rho$ it is possible that multiple subpartitions meet the criteria laid out in (4) of the main document.  This is particularly true when running the algorithm with a relatively small number of MCMC  samples coupled with a diffuse posterior distribution. If, after running the algorithm for a given list $\mathcal{I}$, the CHIPS procedure identifies more than one subpartition with the same number of items and the same posterior probability, this suggests that the number of posterior draws may be insufficient. In most clustering applications the model involves continuous parameters and observations, and the posterior distribution of $\rho$ typically assigns different probabilities to distinct partitions. Thus, with a sufficiently large number of MCMC samples, we generally expect the algorithm to produce a unique solution. In this sense, obtaining multiple solutions provides a useful diagnostic of Monte Carlo error and highlights the need to increase the number of posterior draws.

In addition, the length of the list $\mathcal{I}$ also plays an important role in the stability of the algorithm. When the algorithm is run with several independently generated lists $\mathcal{I}$, we should expect the same result from the CHIPS algorithm. If different outputs arise across these runs, this indicates that the list of starting points is too small and should be enlarged. By appropriately adjusting both the number of posterior draws and the size of the list $\mathcal{I}$, the CHIPS algorithm can be made stable and reliable in practice. We explore the impact the number of posterior samples, number of multiple runs, and $\gamma$ has on the outcome of the CHIPS algorithm in a small simulation study that we now detail.

\begin{figure}
\begin{center}
\includegraphics[scale=0.65]{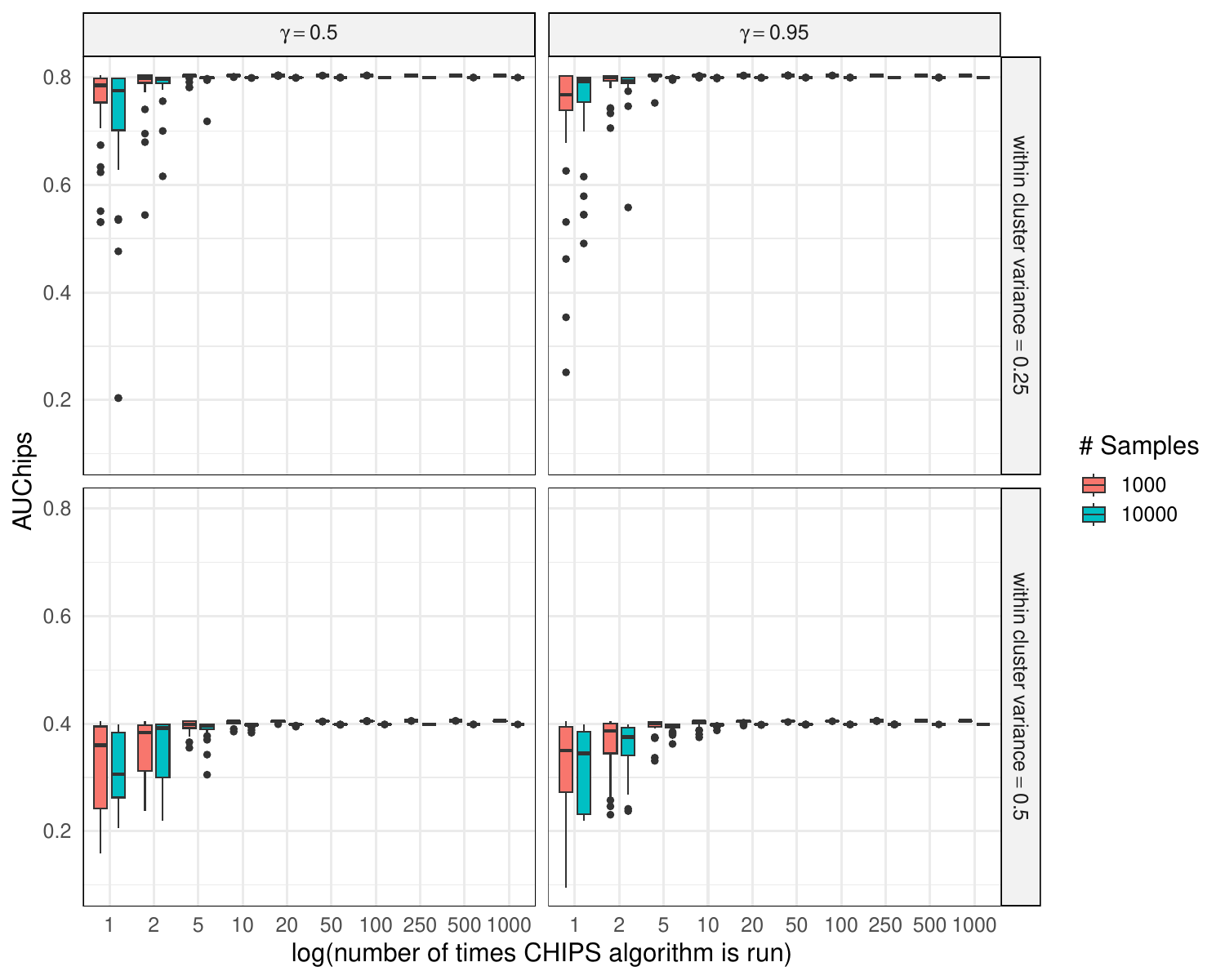}
\end{center}
\caption{Results of running CHIPS algorithm based on synthetic datasets in bottom two rows of Figure 1 found in the main document. The AUChips values from 30 replicated experiments for a number of {\tt nRun} values are displayed}
\label{fig:stability2}
\end{figure}

Using the synthetic datasets displayed in the bottom two rows of Figure 1 in the main document, we collected 1000 and 10,000 MCMC samples for $z$ using the ``z'' model (see Section 5 of the main document).  For both collections of MCMC samples we ran the CHIPS algorithm using {\tt nRuns} $\in \{1, 2, 5, 10, 20, 50, 100, 250, 1000\}$.  This was repeated 30 times so that we had 30 subpartition estimates for each of the {\tt nRuns} values.   We then enumerated the number of unique subpartitions across the 30 runs and also saved the AUChips value for each run. Results are provided in Figure \ref{fig:stability1} and \ref{fig:stability2}.  Notice that when there is relatively little cluster overlap a unique subpartition is identified even with a smaller number of MCMC samples and a smaller list $\mathcal{I}$. However, when there is significant cluster overlap then a list $\mathcal{I}$ with at least 150 items and at least 10,000 MCMC draws were required to obtain a unique partition for $\gamma =0.95$.  The AUChips value seems to stabilize in all scenarios using less MCMC samples and {\tt nRuns} values (see Figure \ref{fig:stability2}).  Overall it seems that the CHIPS algorithm is able to identify a unique subpartition so long as the number of MCMC samples and  the size of $\mathcal{I}$ are sufficient.  The exact number needed to ensure that the algorithm arrives at a unique subpartition is case-specific.

\section{Additional Results from Simulation Studies}

In Figure \ref{fig:simulation_compare_with_salso1} we provide results from the simulation study that considers variation of information loss for $n=400$.  Figures \ref{fig:simulation_compare_with_salso2} and \ref{fig:simulation_compare_with_salso3} display results under Binder loss for $n=100$ and $n=100$ respectively.  The results here follow similar patterns to those displayed in the main document.  The key differences are with $n=400$, there is essentially no difference between SALSO and CHIPS \& SALSO, and AUChips doesn't dip as low under moderate cluster separation.  Under Binder loss, the SALSO point estimate seems to be better under some scenarios than the CHIPS \& SALSO for both $n=100$ and $n=400$.

\begin{figure}
\begin{center}
\includegraphics[scale=0.65, page=5]{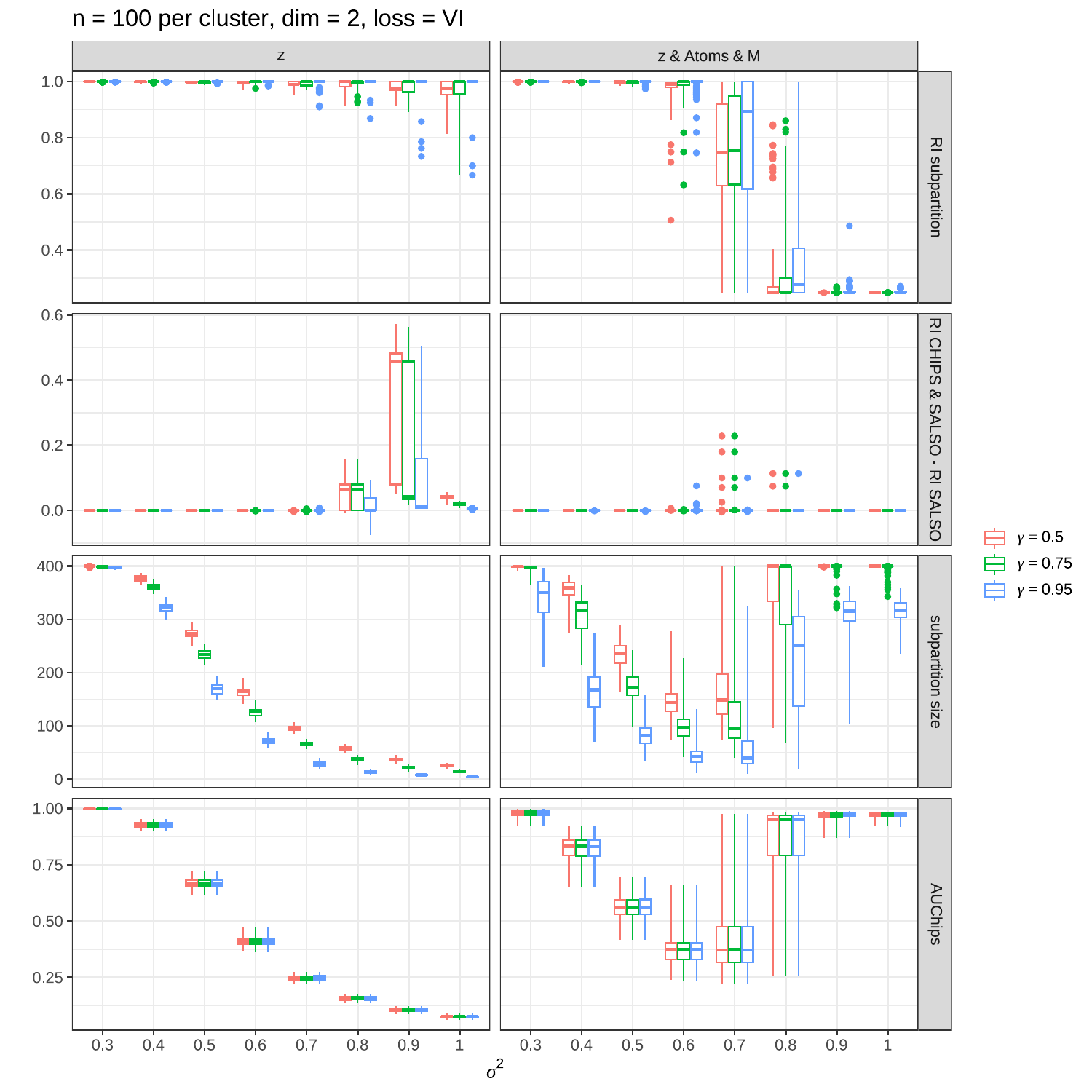}
\end{center}
\caption{Results from the simulation study when variation of information loss is used.  Data are generated from bivariate Gaussians and are composed of four clusters each with 100 observations. The top two rows display results for the Rand index, the third row displays the number of observations included in the subpartition, and the bottom row AUChips values.}
\label{fig:simulation_compare_with_salso1}
\end{figure}

\begin{figure}
\begin{center}
\includegraphics[scale=0.65, page=8]{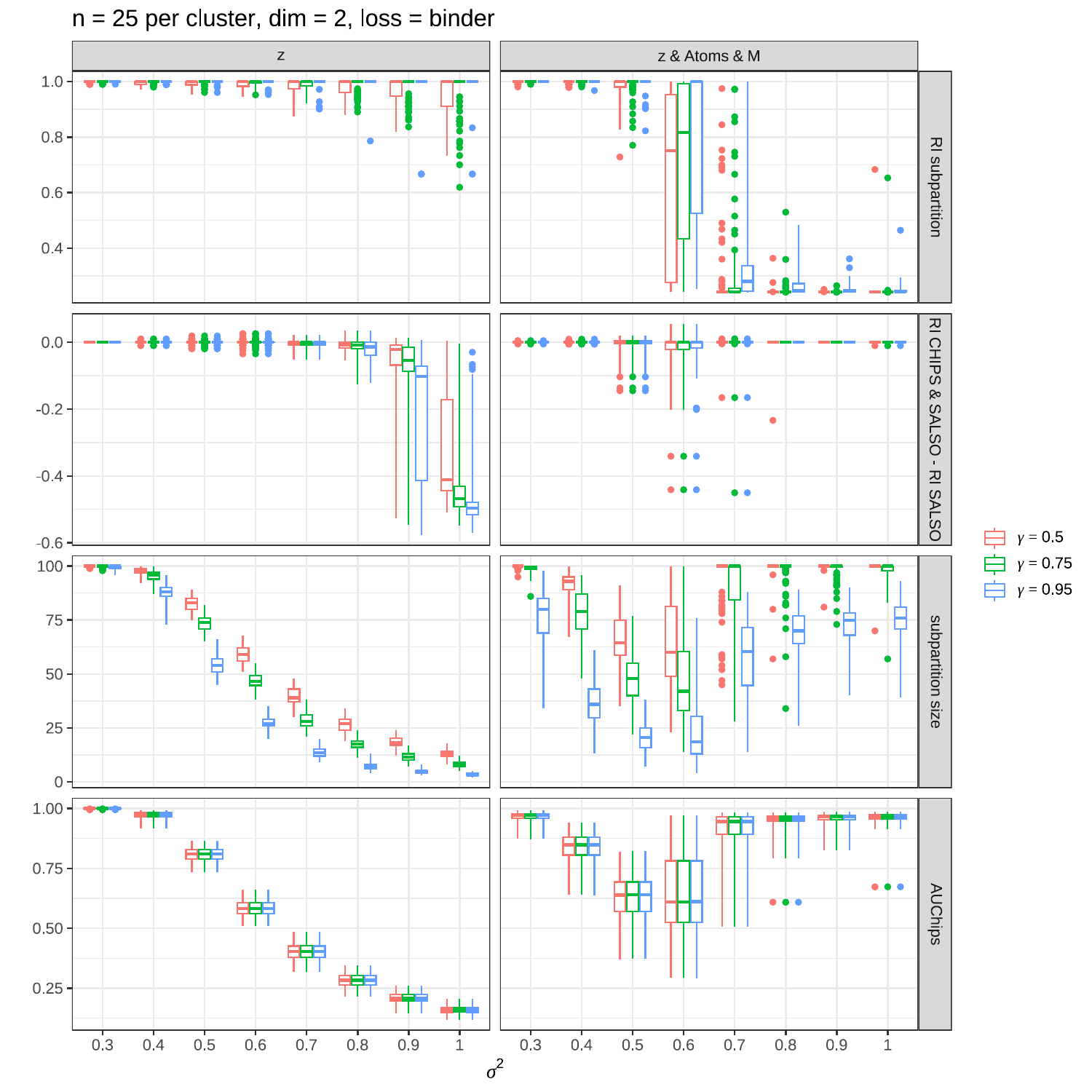}
\end{center}
\caption{Results from the simulation study when Binder loss is used.  Data are generated from bivariate Gaussians and are composed of four clusters each with 25 observations. The top two rows display results for the Rand index, the third row displays the number of observations included in the subpartition, and the bottom row AUChips.}
\label{fig:simulation_compare_with_salso2}
\end{figure}

\begin{figure}
\begin{center}
\includegraphics[scale=0.65, page=9]{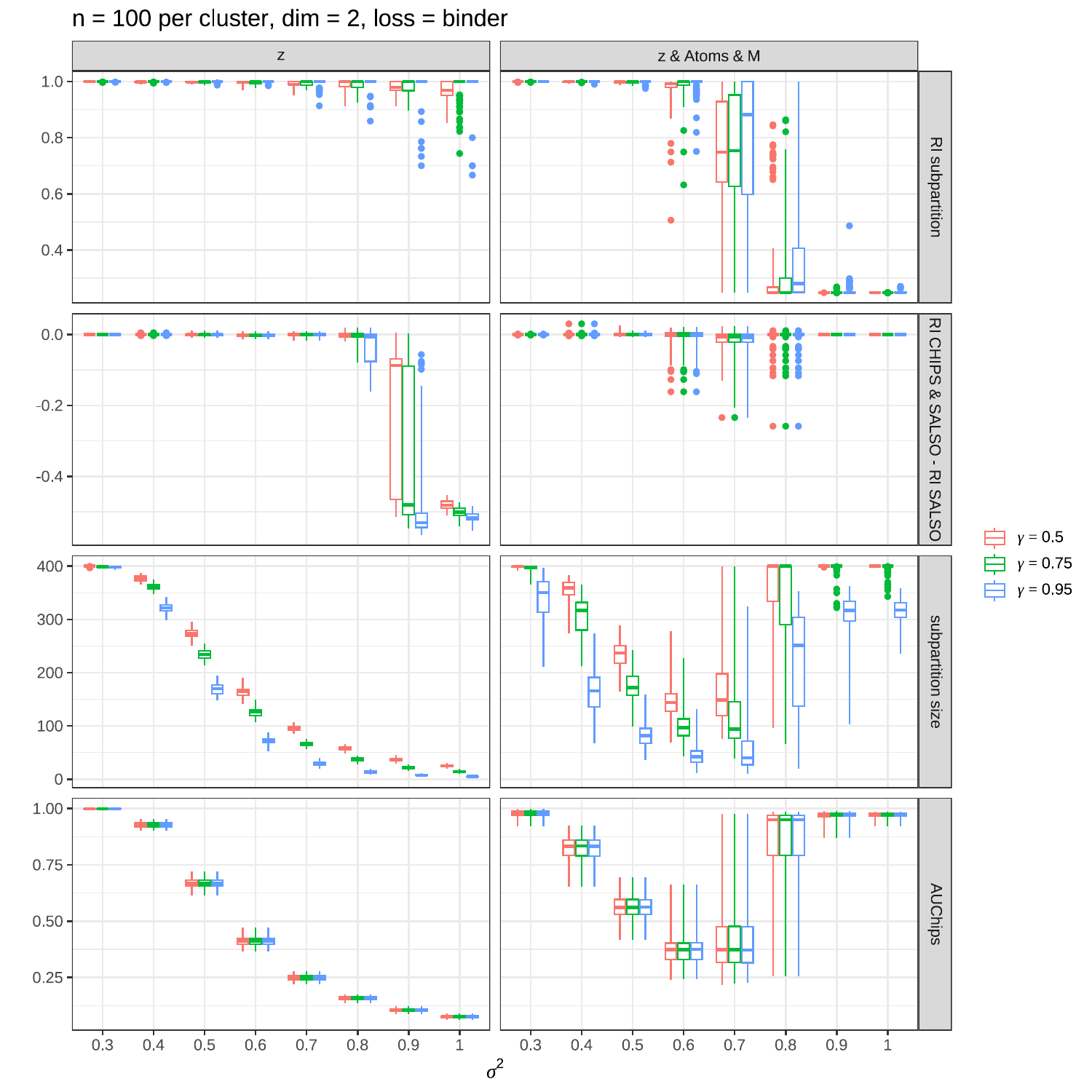}
\end{center}
\caption{Results from the simulation study when Binder loss is used.  Data are generated from bivariate Gaussians and are composed of four clusters each with 100 observations. The top two rows display results for the Rand index, the third row displays the number of observations included in the subpartition, and the bottom row AUChips.}
\label{fig:simulation_compare_with_salso3}
\end{figure}

Figures \ref{fig:simulation_compare_with_wade1} - \ref{fig:simulation_compare_with_wade3} display comparisons between the credible ball of \cite{wade&ghahramani:2018} and the interpretable credible set.  Associations found here are similar to those in the main document as well.  The average Rand index associated with partition in the ball of \cite{wade&ghahramani:2018} and the interpretable credible set are essentially the same, but when only considering the subpartition, it is slightly higher for the interpretable credible set.  The same relationships hold true under Binder loss as well.

\begin{figure}
\begin{center}
\includegraphics[scale=0.65, page=12]{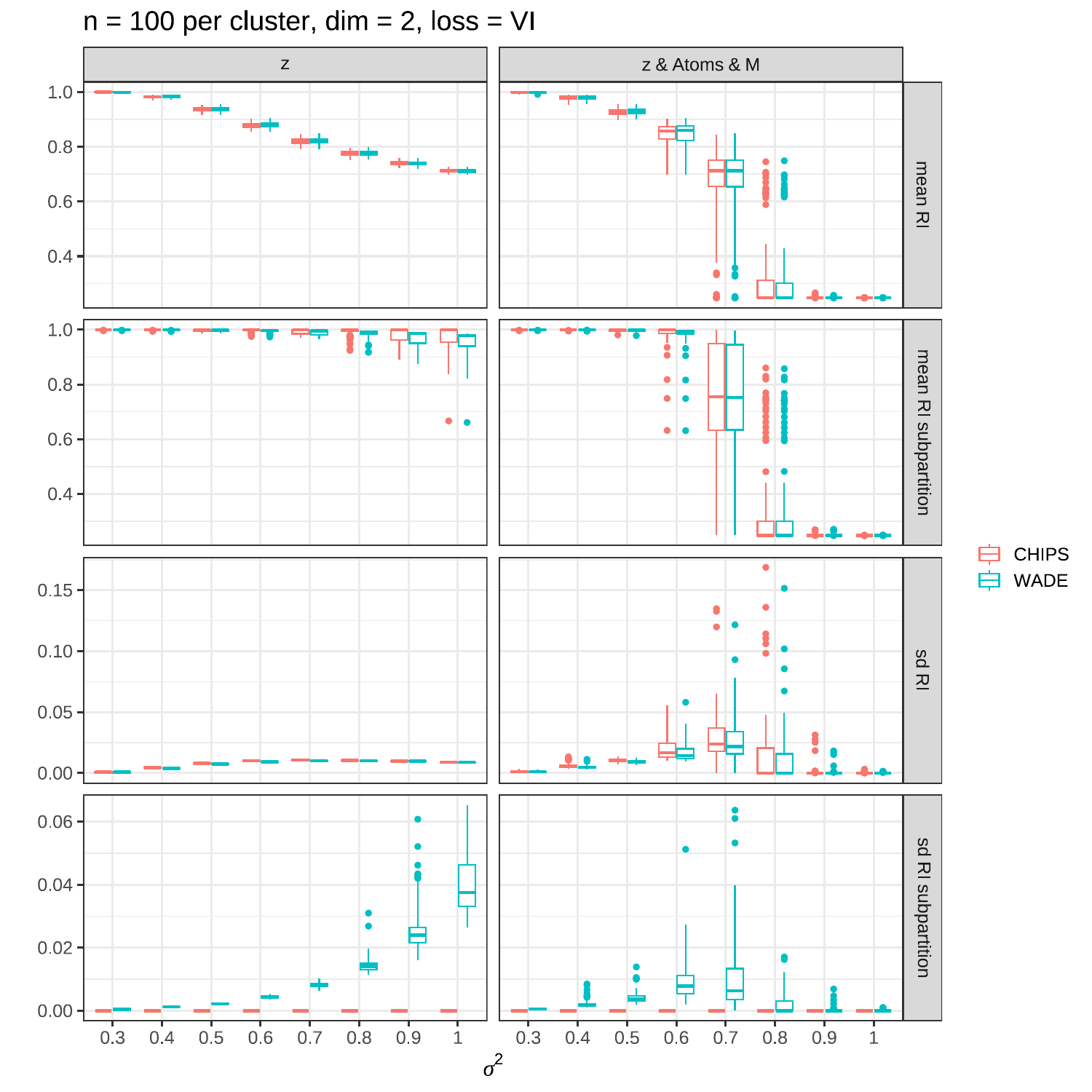}
\end{center}
\caption{Results from the simulation study comparing properties of the interpretable credible set to that of the credible ball of \cite{wade&ghahramani:2018}.  The threshold was set to $\gamma = 0.75$ when forming both credible sets. This is for $n=400$ and variation of information loss.}
\label{fig:simulation_compare_with_wade1}
\end{figure}

\begin{figure}
\begin{center}
\includegraphics[scale=0.65, page=15]{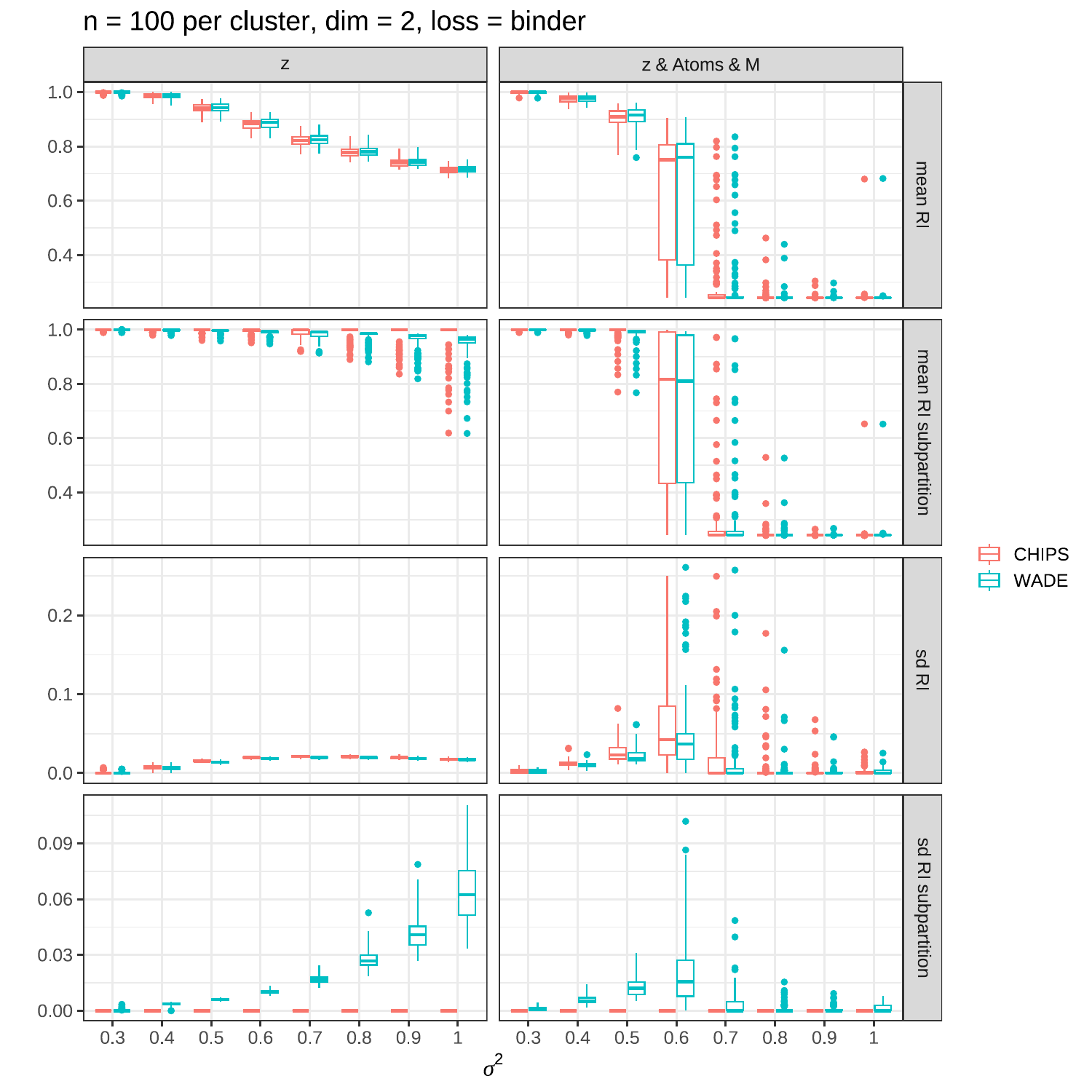}
\end{center}
\caption{Results from the simulation study comparing properties of the interpretable credible set to that of the credible ball of \cite{wade&ghahramani:2018}.  The threshold was set to $\gamma = 0.75$ when forming both credible sets.  This is for $n=100$ and Binder loss.}
\label{fig:simulation_compare_with_wade2}
\end{figure}

\begin{figure}
\begin{center}
\includegraphics[scale=0.65, page=16]{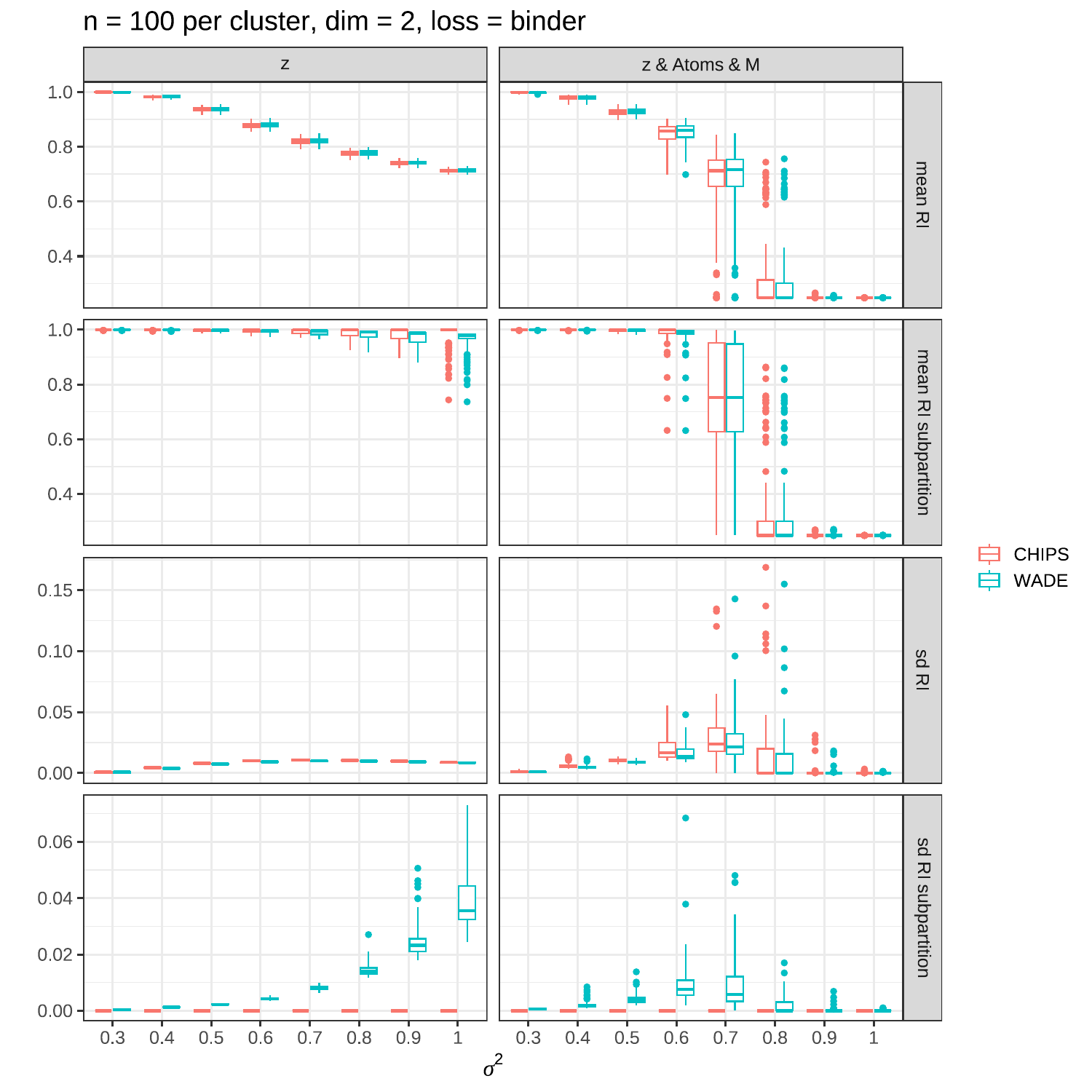}
\end{center}
\caption{Results from the simulation study comparing properties of the interpretable credible set to that of the credible ball of \cite{wade&ghahramani:2018}.  The threshold was set to $\gamma = 0.75$ when forming both credible sets.  This is for $n=400$ and Binder loss.}
\label{fig:simulation_compare_with_wade3}
\end{figure}

\section{Additional Details of Model Used in the Federalist Papers Illustration}
As mentioned in the main document, we employ the model outlined in \cite{chandra_etal:2023}.   and point the reader to the paper complete model details.  What follows are prior specifications that will be understood only after familiarizing yourself with the model.  We set the latent dimension to $d=10$, and center variables with respect to median and scale with respect to the median standard deviations.  We initialize cluster allocations using Kmeans with 15 clusters.  The hyper-prior values we used are the following $as = 1$, $bs = 0.3$,
$a=0.5$, $diag\_psi\_iw=20$, $niw\_kap=1e-3$, $nu=d+1000$,
$a\_dir=.1$, and $b\_dir=.1$.  See the method's accompanying {\tt R}-code for more details.

\end{document}